\documentclass[10pt,journal,compsoc]{IEEEtran}
    
	\usepackage{amssymb}
	\usepackage{graphicx}
	\usepackage{epstopdf}
	\usepackage{gensymb}
	\usepackage{color}
	\usepackage{amsmath}
	\usepackage{enumerate}
	\usepackage{cite}
	\usepackage{comment}
	\usepackage{amsthm}
	\usepackage{enumitem}
	\theoremstyle{definition}

    \usepackage{float}
    \usepackage{subfigure}
	\usepackage{soul}
	
	\usepackage[framemethod=TikZ]{mdframed}
	\usepackage[justification=centering]{caption}
	\usepackage{algorithmicx}
	\usepackage[ruled]{algorithm}
	\usepackage{algpseudocode}
	\alglanguage{pseudocode}
    
    \usepackage{bbding}
    
	\usepackage{array}
	\usepackage{booktabs}
	\usepackage{multirow}

	\usepackage{threeparttable} 
	\usepackage{hyperref}
	\newcommand{\tabincell}[2]{\begin{tabular}{@{}#1@{}}#2\end{tabular}}
	\usepackage{mathtools}

	\usepackage{url}
	\urldef{\mailsa}\path|{alfred.hofmann, ursula.barth, ingrid.haas, frank.holzwarth,|
		\urldef{\mailsb}\path|anna.kramer, leonie.kunz, christine.reiss, nicole.sator,|
		\urldef{\mailsc}\path|erika.siebert-cole, peter.strasser, lncs}@springer.com|

\newtheorem{theorem}{\textbf{\emph{Theorem}}}

\definecolor{garrisonpink1}{rgb}{0.858, 0.188, 0.478}

\newcommand{\mypara}[1]{\vspace{2pt}\noindent\textbf{{#1: }}}
\definecolor{PangLHpink2}{rgb}{0.5, 0.5, 1}

\definecolor{yjr}{rgb}{0.977, 0, 0}

\begin{document}

\title{MUD-PQFed: Towards \textbf{M}alicious \textbf{U}ser \textbf{D}etection in \textbf{P}rivacy-Preserving \textbf{Q}uantized \textbf{Fed}erated Learning}

\author{

Hua Ma\IEEEauthorrefmark{1}, Qun Li\IEEEauthorrefmark{1}, Yifeng Zheng\IEEEauthorrefmark{2}, Zhi Zhang, Xiaoning Liu, Yansong Gao\IEEEauthorrefmark{2} , \\ Said F. Al-Sarawi , Derek Abbott .

\IEEEcompsocitemizethanks{\IEEEcompsocthanksitem \IEEEauthorrefmark{1}H. Ma and \IEEEauthorrefmark{1}Q. Li contributed equally to the study.}

\IEEEcompsocitemizethanks{\IEEEcompsocthanksitem \IEEEauthorrefmark{2} Corresponding authors.}

\IEEEcompsocitemizethanks{\IEEEcompsocthanksitem H.~Ma, S.~Al-Sarawi, and D.~Abbott are with the School of Electrical and Electronic Engineering, The University of Adelaide, Australia. H.~Ma is also with Data61, CSIRO. \{hua.ma;said.alsarawi;derek.abbott\}@adelaide.edu.au}

\IEEEcompsocitemizethanks{\IEEEcompsocthanksitem Q.~Li and Y.~Gao are with the School of Computer Science and Engineering, Nanjing University of Science and Technology, China. \{120106222757;yansong.gao\}@njust.edu.cn}

\IEEEcompsocitemizethanks{\IEEEcompsocthanksitem Y.~Zheng is with the School of Computer Science and Technology, Harbin Institute of Technology, Shenzhen, Guangdong, China.
yifeng.zheng@hit.edu.cn}

\IEEEcompsocitemizethanks{\IEEEcompsocthanksitem X.~Liu is with the School of Computing Technologies, RMIT University, Australia.
xiaoning.trust@gmail.com}

\IEEEcompsocitemizethanks{\IEEEcompsocthanksitem Z.~Zhang is with Data61, CSIRO, Australia. zhi.zhang@data61.csiro.au}

}

\IEEEtitleabstractindextext{		
\begin{abstract}

Federated learning (FL), as a distributed machine learning paradigm, has been adapted to mitigate privacy concerns of clients. Attributing to its advancement, institutions (i.e., hospitals) with sensitive data leverage FL to collaboratively train a global model without transmitting their raw data. Despite attractive, there exists various inference attacks that can exploit the shared plaintext model updates embedding traces of clients private information, causing severe privacy concerns. 
To alleviate such privacy concerns, cryptographic techniques such as secure multi-party computation and homomorphic encryption have been incorporated for privacy-preserving FL. However, this inevitably exacerbates security concerns once clients are malicious to launch attacks, in particular, corrupting the model to ruin the main impetus of benign clients in FL. Those benign clients aim to gain a better global model by contributing their computational and communicational resources and their local data for local model updates. Such security issues in privacy-preserving FL, however, lack elucidation and are under-explored.
This work presents the first attempt towards elucidating the trivialness of performing model corruption attacks against lightweight secret sharing based privacy-preserving FL. We consider the scenario where the model updates are \textit{quantized} to reduce the communication overhead in this case, the adversary can simply provide local parameters out of the small legitimate range to corrupt the model.
We then propose MUD-PQFed, a protocol that can precisely detect malicious clients who performed the attack and enforce fair punishments. By deleting the contributions from the detected malicious clients, the global model utility is preserved as comparable to the baseline global model in absence of the attack. Extensive experiments validate the efficacy in terms of retaining the baseline accuracy and effectiveness in terms of detecting malicious clients in a fine-grained manner.
\end{abstract}

}

\maketitle

\section{Introduction}\label{Sec:Intro}

Federated learning (FL)~\cite{yang2019federated,gao2021evaluation,li2021survey,zhang2022robust} is one of the most popular distributed machine learning techniques. It aims to mitigate the privacy concerns when sensitive data has to be aggregated by a centralized training party to train models in the conventional centralized training paradigm. 
FL collaboratively trains a global model among multiple clients without accessing the local raw data, and thus greatly motivates distributed clients to contribute to gain a global model with better performance than that being trained merely on its limited local data. This is especially the case for the cross-silo FL where all clients will participate in each round of aggregation, instead of only selecting a fraction of clients to participate in each round~\cite{zhang2020batchcrypt,kairouz2021advances}. 
This is particularly suitable for the clients that are institutional organizations (i.e., hospitals) with limited data at local and wish to collaboratively train a more accurate model.

In order to obtain higher accuracy and encourage more clients with high-quality resources (e.g., data, computing resources) to contribute to FL, incentive mechanisms are always taken into FL~\cite{zhan2021survey}.
For example, \cite{kang2019incentive} proposes an effective incentive mechanism combining reputation and contract theory to motivate clients with high-quality data to participate in FL. In the meantime, in order to prevent malicious clients or low-quality clients intentionally degrading the aggregation quality, it is also important to punish these clients. Therefore, incentive schemes for FL should not only reward honest clients, but also penalize malicious clients once being identified~\cite{gao2022fgfl}.

\vspace{0.2cm}
\noindent{\bf Privacy Inference Attack.} In each round of the FL training, each client independently updates the local model trained over its local data, and transmits the updated local model to the central server who serves as a coordinator. The server aggregates all local models and updates the global model through, e.g., \textsf{FedAvg}~\cite{mcmahan2017communication,gao2020end}. Albeit the FL can greatly mitigate privacy leakages from the raw data, it is still vulnerable to advanced privacy inference attacks such as membership inference attack~\cite{nasr2019comprehensive}, data inversion attack~\cite{geiping2020inverting}, property inference attacks~\cite{ganju2018property} and preference profiling attacks~\cite{zhou2022ppa} when the local model communicates in plaintext with the server.

\vspace{0.2cm}
\noindent{\bf Communication Overhead.} Generally, a local model's parameters are downloaded/uploaded in full-precision by default, which results in large communication overhead to the FL. The overhead is further exacerbated in a scenario where a large number of clients are involved, and each local model contains millions or even tens of millions of parameters, rendering somehow unbearable communication overhead. 

A practical solution to the excessive communication overhead caused by full-precision parameters in FL is to quantize the parameters, namely quantized FL. As such, ternary FL~\cite{xu2020ternary} and ultimately binary FL~\cite{yang2021communication} have recently been proposed to substantially reduce the communication overhead. In addition, \textit{an inadvertent security merit of the quantized FL is that it has constrained the legitimate parameter change degree} if an malicious client does tamper the parameters \textit{in plaintext}. For example, tenary parameter can only have three legitimate values \{-1, 0, 1\}, which restricts the legitimate parameter manipulation degree.

\vspace{0.2cm}
\noindent{\bf Security Dilemma of Cryptographic Privacy-Preserving FL.} In order to fundamentally protect clients' privacy (e.g., preventing privacy inference attacks), cryptographic designs have been proposed in which model parameters are encrypted into ciphertext and the aggregation is also performed over ciphertexts to prevent direct access to  local model parameters in plaintext, no matter whether the FL is in full-precision ~\cite{xu2019verifynet,hao2019efficient} or quantized~\cite{wen2017terngrad}.

However, such cryptographic aggregation designs render security attacks much easier, e.g., corrupting the global model convergence to cause denial of service  (DoS). As the privacy-preserving FL designs are merely focusing on privacy but not security attack perspective, the security dilemma of the privacy-preserving FL is much less elucidated (see details in Section~\ref{sec:FL}). 

\vspace{0.2cm}
\noindent{\bf Privacy-Preserving Quantized FL.} Although it greatly reduces communication overhead, it is inevitably threatened by security attacks, resulting in worrisome deployments. This is exacerbated by the fact that the malicious clients can easily corrupt the model without obeying the legitimate parameter range.
 
Consider ternary FL as an example,
since the parameters are values from \{-1, 0, 1\}, malicious clients only need to make some illegitimate changes, such as changing -1 to -5, which will affect the correctness of the aggregated model such as preventing the global model from converging. 

It is imperative to identify malicious clients in the quantized privacy-preserving FL, then remove and penalize them to ensure the utility of the global model without falsely punishing normal clients. However, there is no existing robust privacy-preserving FL scheme that has considered such \textit{trivial but realistic model corruption} attack caused by illegitimate parameter unique to \textit{quantized FL}.

\vspace{0.2cm}
\noindent{\bf Our Contributions.} In this work, we propose the first protocol, coined as MUD-PQFed, that is able to accurately detect malicious clients who corrupt the global model by submitting illegitimate parameters~\cite{fang2020local} to their local models. MUD-PQFed works for the quantized privacy-preserving FL where the local model parameters are in the form of ciphertexts. In addition,  MUD-PQFed can efficiently mitigate the attack effects by removing the detected malicious clients' poisoned (alike fault injection to some extent) local models. The main contributions/results of our work are summarized as follows:
  \begin{itemize}
    \item We have revealed through experiments that for the quantized FL, an attack to corrupt a global model by local model poisoning ~\cite{fang2020local}) in the ciphertext form is trivial to succeed. A small change (e.g., the number of model parameters or/and the degree per parameter) will have a substantial impact on the accuracy of the global model, and thus make it completely corrupted.
    \item We have then proposed MUD-PQFed, the first lightweight malicious client detection scheme for quantized FL in ciphertext building upon secret sharing with a single-server setting. Our scheme groups clients based on the hypermesh to correctly identify malicious clients, and then delete their adverse contributions to retain the baseline utility of the global model.
    \item We have conducted extensive experiments, validating that MUD-PQFed can effectively detect malicious clients (i.e., 100\% true positive rate), thus ensuring the accuracy of privacy-preserving quantized FL comparable to its baseline accuracy without the attacks. The capability of correctly identifying malicious clients can complement the incentive mechanism so as to punish those clients in privacy-preserving quantized FL.
  \end{itemize}

\mypara{Paper Organization} Section~\ref{sec:background} provides necessary backgrounds. Preliminaries of hypermesh and secret sharing for building the MUD-PQFed are given in Section~\ref{sec:preliminary}. Section~\ref{sec:security Attack} presents security attacks, particularly, model corruption attacks, on privacy-preserving quantized FL. Section~\ref{sec:system model} elaborates on the overall design of MUD-PQFed. Extensive experiments are conducted in Section~\ref{sec:experiments} to validate the effectiveness and efficiency of the MUD-PQFed. Further discussions are made in Section~\ref{sec:discussion}, followed by the conclusion in Section~\ref{sec:conclusion}.

\section{Background}\label{sec:background}

\subsection{Federated Learning}\label{sec:FL}

F\textbf{}ederated learning (FL)~\cite{li2014scaling,konevcny2016federated} is a distributed machine learning technique. It enables $n_c$ distributed clients to collaboratively train a global model. In each round, client $i$ with $i \in \{1,..., n^i_c\}$ uses its own local private dataset $\mathcal{D}_i$ for local model training, and sends the obtained gradient ${\bf w}_i$ to the server who acts as the coordinator. 
After receiving the gradients from all clients, according to the most used \textsf{FedAvg}~\cite{konevcny2016federated}, the server aggregates them and updates the global model parameters as:
\begin{equation}
    \tilde{\bf w}=\sum_{i=1}^{n_c}{\frac{|\mathcal{D}_i| }{|\mathcal{D}|}}\textbf{w}_i,
\end{equation}
where $|\mathcal{D}|=\sum\limits_{i=1}^{n_c}|\mathcal{D}_i|$ with $|\cdot|$ the cardinality. Then the server sends $\tilde{\bf w}$ to all $n_c$ clients for next round. Although the server is not allowed to access the client's private data, it is still possible to infer clients' sensitive information through inference attacks ~\cite{nasr2019comprehensive,luo2021feature}. 

\vspace{0.2cm}
\noindent{\bf Privacy-Preserving FL.} Commonly used techniques to further protect privacy in FL include differential privacy (DP)~\cite{geyer2017differentially,mugunthan2020privacyfl}, homomorphic encryption (HE)~\cite{zhang2020batchcrypt,ma2022privacy,hardy2017private} and secure multi-party computation (MPC)~\cite{xu2019hybridalpha}. The DP is less effective to protect privacy compared with the other two cryptographic techniques. In addition, DP suffers (unbearable) accuracy drops.

Therefore, this paper mainly focuses on cryptography-enabled privacy-preserving FL, where the local model parameters are in the ciphertext form, which is illustrated in Fig.~\ref{fig:flmodel}. Generally, the cryptographic means can be divided into two classes: multi-party computation based, and homomorphic encryption based. The former usually uses secret sharing, while the latter uses fully homomorphic encryption (FHE) or partially homomorphic encryption (PHE).

\vspace{0.2cm}
\noindent\textit{HE based.} Hao \textit{et al.}~\cite{hao2019efficient} used FHE to achieve secure aggregation of gradients. The local gradient vector is first perturbed using a distributed Gaussian mechanism. Then, the perturbed gradient vector is encrypted into the BGV ciphertext (called the inner ciphertext).
Finally, the ciphertext encrypted by BGV is embedded in the augmented learning with error (A-LWE) ciphertext (called external ciphertext) to realize the secure FL aggregation protocol. 
Ma \textit{et al.}~\cite{ma2022privacy} proposed a multi-key homomorphic encryption scheme that is resistant to attacks from curious internal actors as well as external malicious adversaries while achieving accuracy protection. In this scheme, the client encrypts the gradient using the aggregated public key. To get the plaintext from this encrypted sum, the server asks all clients to calculate their decrypted shares, and finally, the server merges them with the ciphertext to decrypt the encrypted sum of all gradients. However, applying HE to FL incurs substantial computational and communicational overhead~\cite{fereidooni2021safelearn}.

\vspace{0.2cm}
\noindent\textit{MPC based.} As a scheme that can effectively protect privacy without causing too much extra overhead, secure multi-party computation is increasingly used in FL. Xu \textit{et al.} ~\cite{xu2019hybridalpha} proposed a protocol that is simple, efficient, and resilient to withdrawing participants by applying MPC to FL. In the Chain-PPFL scheme proposed by Li \textit{et al.}~\cite{li2020privacy}, the client is organized into a chain structure. The server sends a random number to the first client in each chain. The first client in each chain adds the random number as a mask to its gradient, and then randomly sends the result to a neighbor. The neighbor client encrypts its gradient as a mask. Finally, the last client sends the final result to the server.

\begin{figure}[t!]
    \centering
    \includegraphics[width=0.4\textwidth]{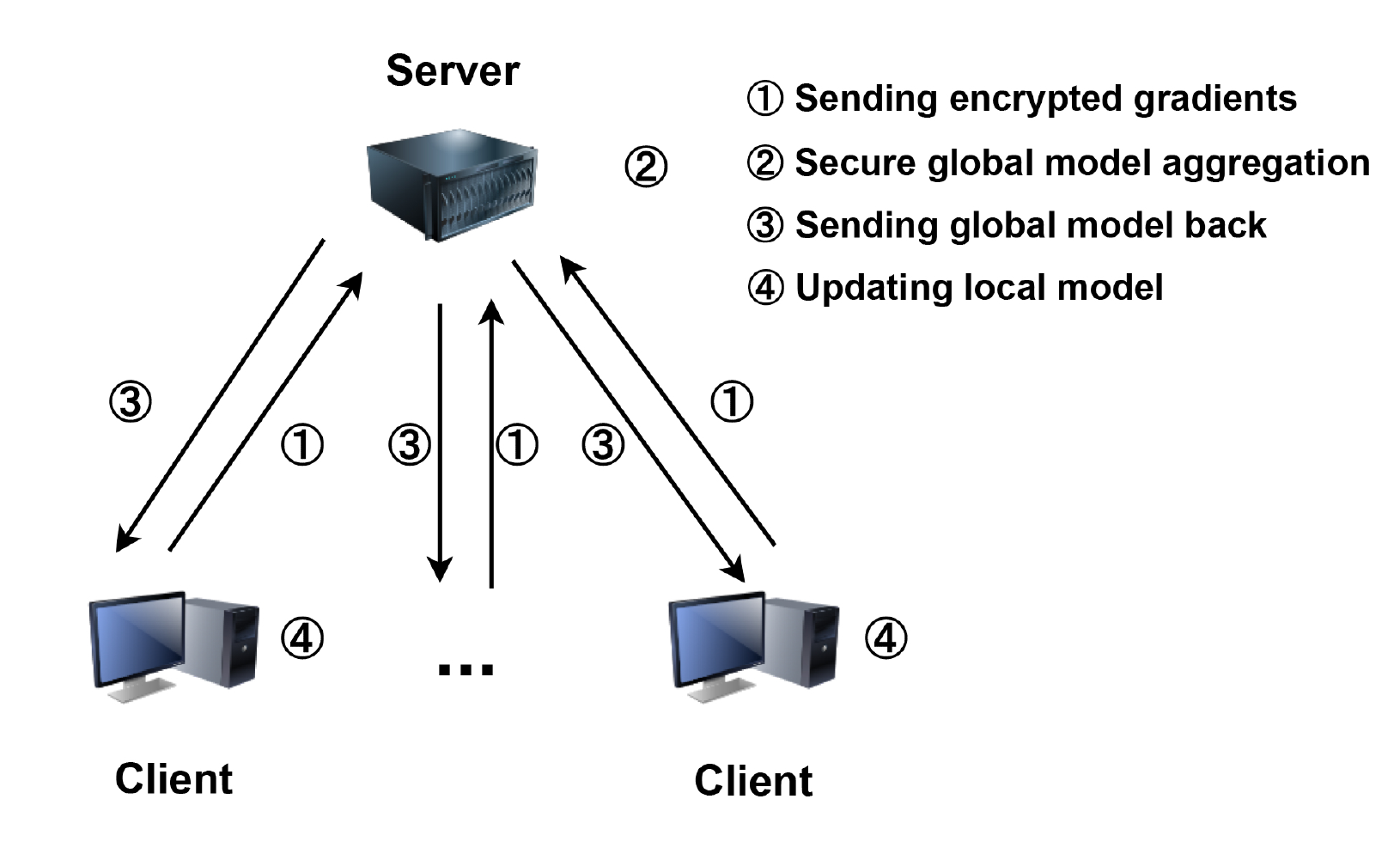}
    \caption{FL aggregation in ciphertext.}
    \label{fig:flmodel}
\end{figure}

\vspace{0.2cm}
\noindent{\bf Security Robustness.} There are merely few attentions paid to intrinsic security threats under the privacy-preserving FL. Bhagoji \textit{et al.}~\cite{bhagoji2019analyzing} demonstrated that targeted model poisoning against FL is also effective for models with Byzantine-resilient aggregation rules. 
Jiang \textit{et al.}~\cite{jiang2022comprehensive} showed that for vertical partition FL, private set intersection cannot resist reconstruction attacks; as a defense, when the DP privacy budget is large, it cannot effectively resist attacks.
Xu \textit{et al.}~\cite{xu2020privacy} proposed a scheme to reduce the influence of irregular clients on training results in ciphertext state. The scheme first calculates the reliability of each client, and then determines the proportion of the client in the final aggregated result according to the reliability. It can be seen as a method of resisting poisoning attacks in the state of ciphertext to some extent, but this scheme incurs high computational overhead due to the usage of HE and is incapable of identifying irregular clients. 

So \textit{et al.}~\cite{so2020byzantine} proposed BREA. Each client locally computes the pairwise distances between locally updated secret shares belonging to other clients, and sends the computed results to the server. Then the server collects computations from a sufficient number of clients, recovers pairwise distances between local updates, and performs client selection for model aggregation. However, this paper points out that relying on distance-based outliers to find malicious clients makes it difficult to distinguish whether large distances between local updates are due to non-IID data distributions or Byzantine attacks. In addition, the overall FL computational complexity of gaining these pairwise distances is up to \textsf{O}($n_c^2$) with $n_c$ the number of participated clients.
Liu \textit{et al.}~\cite{liu2021privacy} proposed the first scheme to detect poisoning under HE enabled ciphertext conditions. The scheme uses the computed median over ciphertext as a benchmark to calculate the Pearson correlation between other local model's parameters and the median to identify and exclude potential malicious clients.

Ma \textit{et al.}~\cite{ma2022shieldfl} proposed ShieldFL, a scheme that can detect poisoning attacks in the ciphertext state. The ShieldFL is based on secure cosine similarity on encrypted local gradients to distinguish malicious clients from normal clients. However, both \cite{liu2021privacy} and \cite{ma2022shieldfl} are based on privacy protection under HE, which will undoubtedly lead to very large overhead (i.e., violates the motivation of reducing communication overhead of quantized FL focused in this work) and the scalability for the number of clients is not high. In addition, even if malicious clients can be detected, both schemes leak other information about clients (such as Pearson correlation coefficient and cosine similarity), which may lead to other potential attacks such as data inversion attack and model stealing attack recognized in~\cite{liu2021privacy}. Moreover, these two studies~\cite{ma2022shieldfl,liu2021privacy} build upon two (non-colluding) servers, which interact with clients and communicate with each other to carry out a secure two-party protocol. Such a two-server setting is different from the common one-server FL setting.

In summary,  privacy-preserving (quantized) FL with attack robustness is still under-explored. More specifically, those solutions~\cite{so2020byzantine,ma2022shieldfl,liu2021privacy} are inapplicable for detecting malicious clients that submit illegitimate local model parameters in the quantized FL, which can easily render global model corruption.

\subsection{Quantized FL}

In FL, communication overhead is an important issue that could be a severe bottleneck if the available bandwidth is restricted. This is because in FL, the number of clients could be very large, such as millions of smartphones. When full-precision parameters are used to transmit gradient information in the communication process, the large number of clients will incur unbearable communication overhead. 
Besides, a large number of parameters are communicated between the server and the client, but the throughput of the communication channel is limited~\cite{mcmahan2017communication,li2020federated,chen2020joint}. Moreover, as aforementioned, in order to protect the privacy of the FL aggregation process, some schemes adopt homomorphic encryption ~\cite{liu2020secure,cheng2021secureboost} or secure multiparty computation (SMC)~\cite{mohassel2018aby3,mohassel2017secureml} for privacy protection, which exacerbates the communication overhead. 

Therefore, in order to reduce communication overhead and alleviate channel bandwidth pressure and throughput, quantized FL emerges~\cite{reisizadeh2020fedpaq,xu2020ternary,yang2021communication}. The Fedpaq protocol proposed by Reisizadeh \textit{et al.}~\cite{reisizadeh2020fedpaq} allows the clients to quantize their gradients first, and then send it to the server, which reduces communication overhead with an accuracy trade-off. The T-FedAvg protocol proposed by Xu \textit{et al.}~\cite{xu2020ternary} quantizes the gradient to -1, 0 and 1 through a layer-by-layer quantization method, which greatly reduces the communication overhead while maintaining high accuracy. Yang \textit{et al.}~\cite{yang2021communication} proposed the first binary FL framework. They utilize maximum likelihood estimation and auxiliary real-valued parameters to update gradient information. So they can enable binary state learning and reducing communication costs. Accuracy drop due to parameter binarization can be compensated by mixed-precision FL.

\section{Preliminaries}\label{sec:preliminary}
This section provides preliminaries of secret sharing and hypermesh, which will be used to construct MUD-PQFed.

\subsection{Secret Sharing}

Secret sharing is a technique in which the secret value can only be obtained by knowing all secret shares or a number of shares above a pre-defined threshold: the secret is divided into multiple shares~\cite{beimel2011secret}. More specifically, a secret $s$ is divided into $n$ unrelated parts of information, each part of the information is, namely, a sub-key, which is held by a different participant. The secret $s$ can only be recovered when at least $k$ sub-keys/shares are possessed. This scheme is $(k, n)$-secret partition threshold scheme, and $k$ is called the threshold value of the scheme. Secret sharing schemes are often incorporated within the multi-party secure computations.

\subsection{Hypermesh}
The grouping method of FL clients in our MUD-PQFed is built upon hypermesh~\cite{szymanski1995hypermeshes}. A $d^n$ hypermesh (where $n_c$ = $d^n$, $d \ge 2$, $n \ge 2$) consists of $n_c$ nodes. Each node is assigned an $n$-digit identifier $a_{n-1}a_{n-2}...a_0$ such that $a_i$ $\in$ $[0, d)$ for all $0 \leq i < n$. Two nodes are called neighbors if their identifiers differ by one bit. Nodes are connected by $d$-edges (i.e., edges with $d$ endpoints). That is to say, in a hypermesh, an edge has $d$ nodes, one node belongs to $n$ different edges at the same time, and there are $l = \frac{n_c \times n}{d}$  edges in the whole hypermesh. We give an example of a hypermesh in Fig. \ref{fig:hypermesh}. Each client is assigned an identifier by the server, for example the identifier for client 0 is '00'. Each edge corresponds to a group, and the hypermesh shown in Fig. \ref{fig:hypermesh} consists of a total of 16 clients and 8 groups.

\begin{figure}[t!]
    \centering
    \includegraphics[width=0.35\textwidth]{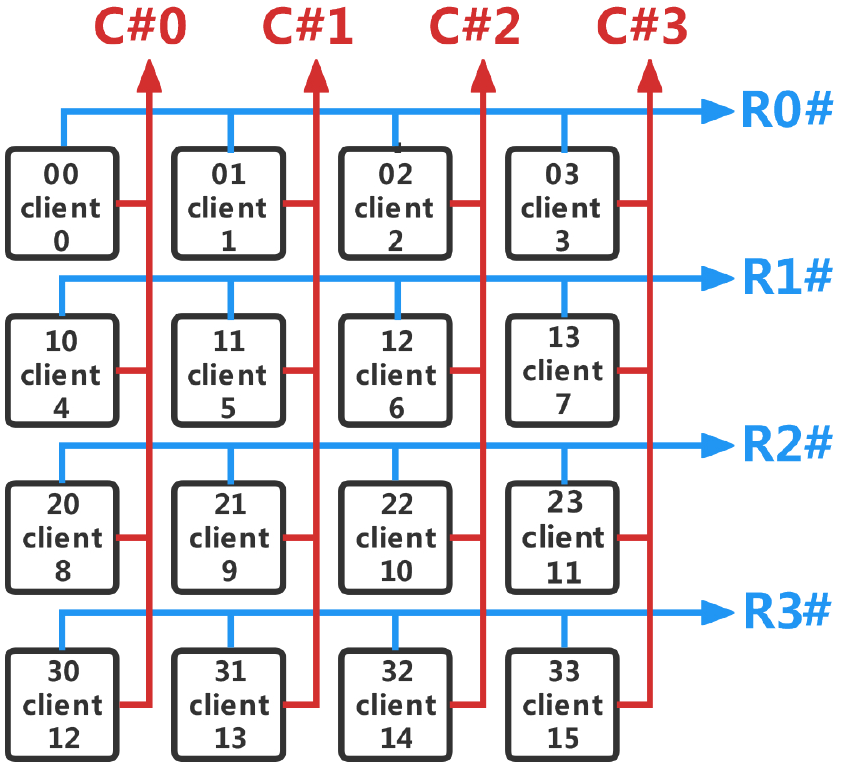}
    \caption{\centering An exemplified $d^n$-hypermesh, with $d=4$ and $n=2$.}
    \label{fig:hypermesh}
\end{figure}

\section{Security Attack on Privacy-Preserving Quantized FL}\label{sec:security Attack}

\begin{figure}[t!]
    \centering
    \includegraphics[width=0.45\textwidth]{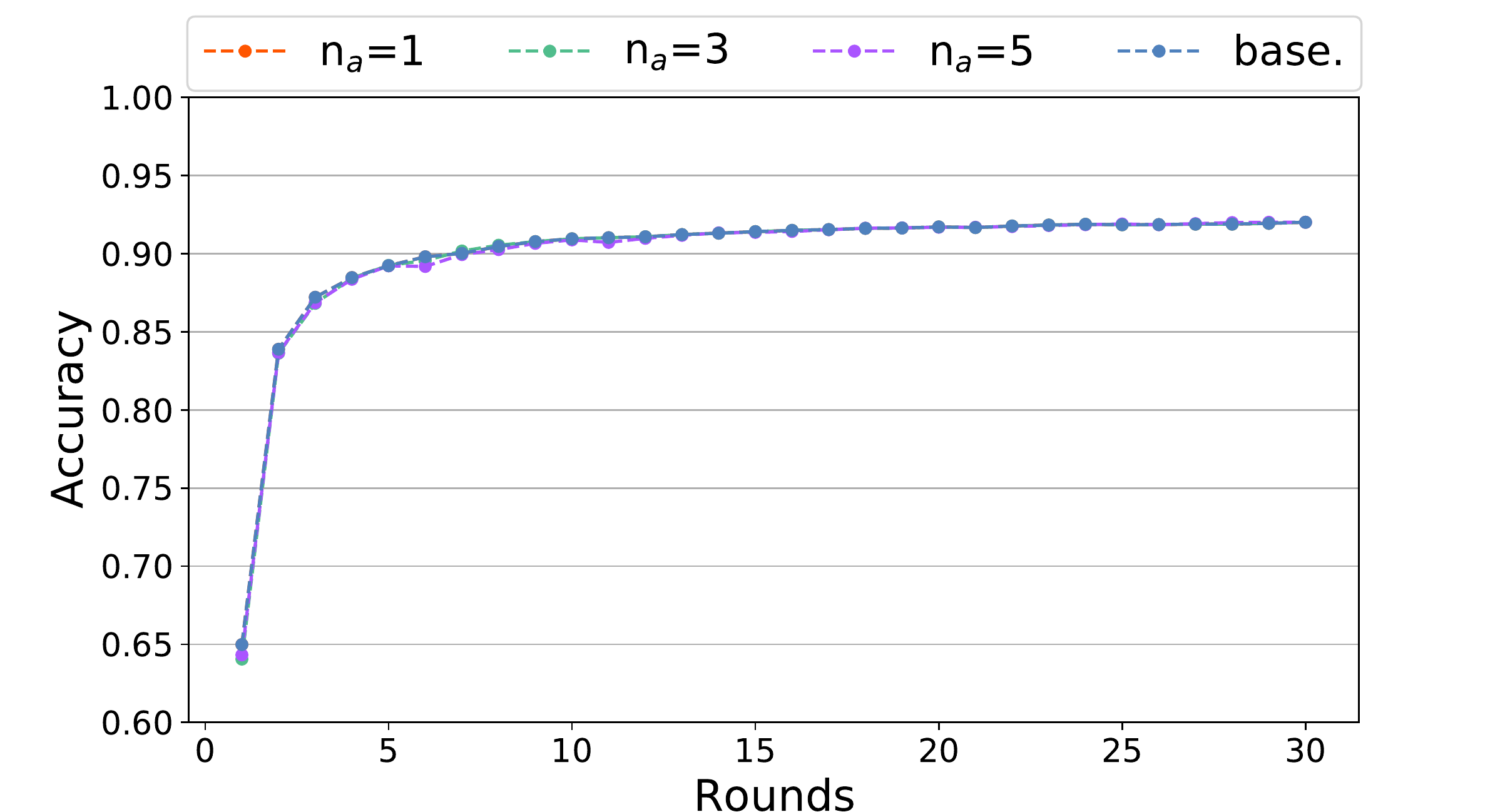}
    \caption{\centering Malicious clients randomly tamper parameters within legitimate range \{-1, 0, 1\}. The MNIST dataset is used.
    $n_a$ is the number of malicious clients.
    }
    \label{fig:range}
\end{figure}

\begin{figure}[t!]
    \centering
    \includegraphics[width=0.45\textwidth]{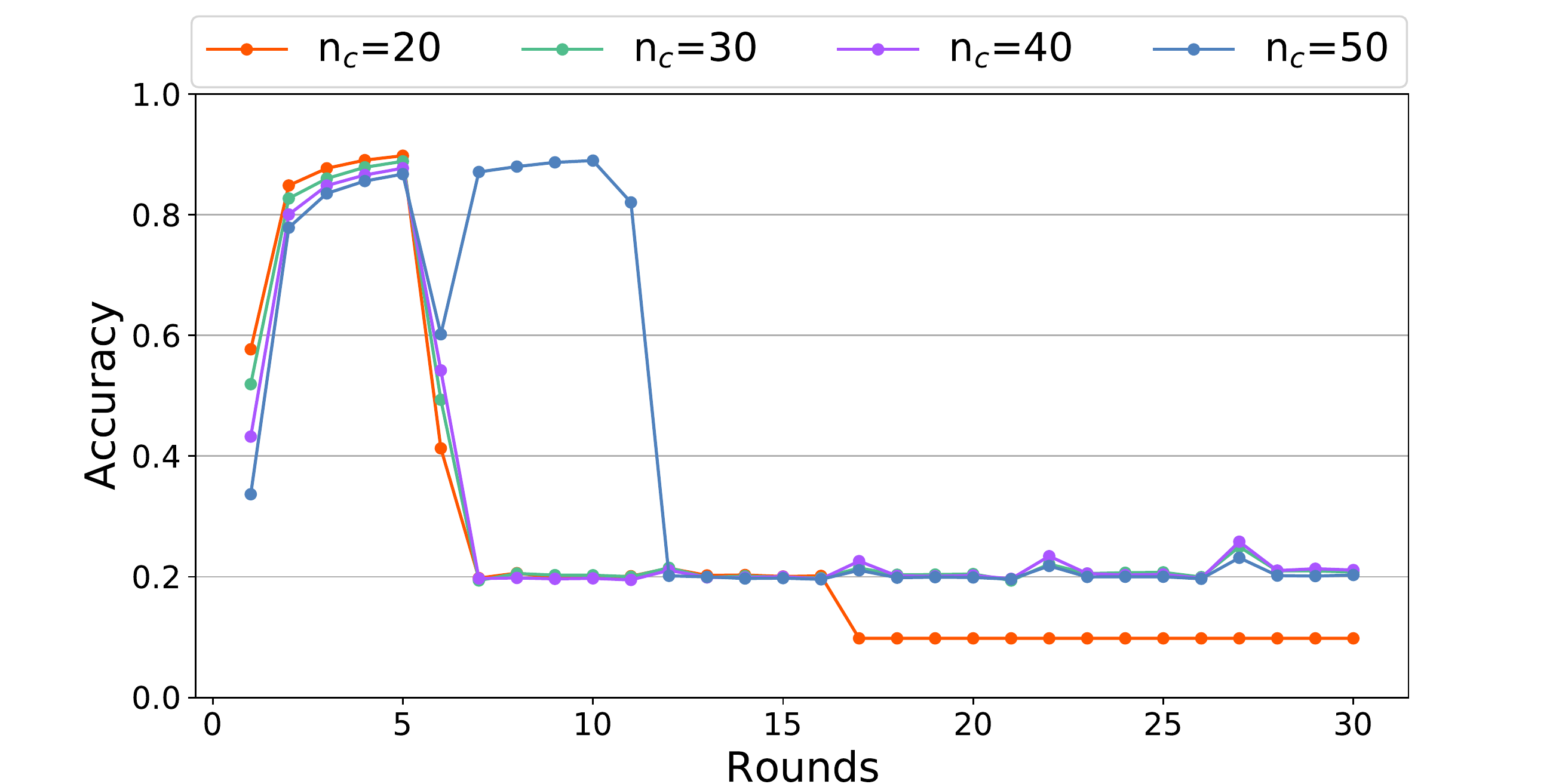}
    \caption{\centering Model corruption attack under ternary FL with only a malicious client---the worst-case from attacker perspective. The MNIST dataset is used. $n_c$ is total number of clients in the FL.}
    \label{fig:none}
\end{figure}

Quantized FL itself can substantially reduce the communication overhead and privacy leakage to some extent~\cite{qiang2021defending} in comparison with the full-precision FL counterpart. In addition, secure computation techniques can complement the quantized FL to fundamentally mitigate the privacy leakages. However, the security threats rooted in the cryptography based privacy-preserving quantized FL is not explored and elucidated. We here show that quantized FL unique model corruption attack via tampering the local model parameter in an illegitimate range is realistic and trivial. Such manipulation, unless carefully treated, can be easily achieved in the ciphertext domain.
Therefore, it is imperative to design a protocol to detect malicious clients in the quantized FL under ciphertext form, while retaining the lightweight feature of the quantized FL. Later we devise MUD-PQFed in Section~\ref{sec:system model} to cater for this critical demand.

\vspace{0.2cm}
\noindent{\bf Model Corruption Attack.} We conducted experiments on a ternary FL scheme in plaintext to demonstrate the trivialness of launching a global model corruption attack by simply tampering with the model parameters to be out of legitimate range. This is exacerbated by the fact that the attacker has been relaxed parameter manipulation range/degree. 

In the experiment, the ternary FL uses the MLP network (i.e., a series of fully connected layers) for training, with 24,330 parameters in total, following~\cite{xu2020ternary}. 
In Fig.~\ref{fig:range}, we show the accuracy when malicious clients can \textit{only make legitimate changes} in parameters within $\{-1,0,1\}$. The total number of clients is $n_c = 25$, and malicious clients change parameters in the $6_{\rm th}$, $11_{\rm th}$, $16_{\rm th}$, $21_{\rm st}$, and $26_{\rm th}$ rounds, respectively. We can see that even when the number of malicious clients is up 20\% ($n_a=5$), merely providing random legitimate parameters for a given layer of 600 parameters has no attack effect on corrupting the global model accuracy.  

We now assume that there is only one malicious client---worst-case from the attacker's perspective, and increase the value of same 600 parameters by (adding) 50 per parameter (i.e., out of legitimate ranges) in the same $6_{\rm th}$, $11_{\rm th}$, $16_{\rm th}$, $21_{\rm st}$, and $26_{\rm th}$ rounds, respectively. 
Fig.~\ref{fig:none} depicts the experimental results of the corruption attack. Here, {$n_c$} is the number of clients participating in the quantized FL.
It can be seen from the experimental results that global model received by rest clients (i.e., local model from the client's perspective) will have a large accuracy drop after this malicious client submits malicious values, and the smaller the number $n_c$, the greater the impact. 
Obviously, the more clients there are, the smaller the impact a single malicious client can cause, so the accuracy will decrease slowly. Similarly, a malicious client in the ciphertext state can perform the same operation to compromise the global model and evade the illegitimate range check.

\vspace{0.2cm}
\noindent{\bf Hardness of Malicious Client Detection.} Intuitively, it seems that for the quantized FL that only uses secure multi-party computation to achieve secure aggregation to output a global model, the server can determine whether each aggregated model parameter is within a reasonable threshold, and then determine whether there are malicious clients. However, this may not easily hold in practice.

For example, for binary FL, the parameter value is either 0 or 1. If 100 clients participate in the quantized FL and each client submits 1000 parameters, then the aggregated result should be in [0, 100] for each parameter. 
If the quantized FL is in plaintext, the server can determine whether there are malicious clients by checking whether the local parameter range is in [0, 100] per parameter per client. But this per parameter per client check is prohibited in ciphertext state. In addition, the attacker can hide the maliciously injected value within the normal range of [0, 100], e.g., by estimating the tampered parameter firstly. More specifically, the attacker can tamper with the parameter that has a higher chance to be 0 sent by majority of clients. In this context, the aggregated parameter will be still within legitimate range, rendering infeasibility of detecting malicious behavior. Even though the aggregated value does exceed normal ranges, the server can only determine that there are malicious behaviors, but cannot identify which client is malicious because the server can only examine the aggregated value from all clients in the MPC (in particular, using secret sharing) enabled privacy-preserving quantized FL.

As for quantized FL that uses homomorphic encryption for FL privacy protection, it is more vulnerable to such security attacks. Because the server cannot know the plaintext value of the aggregated value opposed to the MPC-based FL aggregation, thereof, it is infeasible to tell whether there are malicious clients relying on the aggregated values in ciphertext. For this reason, when designing MUD-PQFed, we choose to use MPC rather than HE, as the latter exacerbates the hardness of malicious client detection in addition to its higher computational and communicational overhead.

\section{The Design of MUD-PQFed}\label{sec:system model}
\subsection{Threat Model}
We consider a typical cross-silo FL~\cite{zhang2020batchcrypt,kairouz2021advances}, where all clients participate in all FL rounds once they decide to contribute to the global model training. In the MPC-based privacy-preserving FL, there is a server and $n_c$ clients with their local private data. There are $n_a$ out of $n_c$ clients being malicious clients (i.e., controlled by an attacker) who attempt to corrupt the global model by submitting malicious values.  Following related studies~\cite{xu2019hybridalpha,xu2019verifynet}, the rest clients as well as the server are honest but curious. That is, they will honestly follow the protocol specification but may try to infer additional sensitive information.

\vspace{0.2cm}
\noindent{\bf Goals.} The goal of the proposed MUD-PQFed is to detect malicious clients if they are performing the model corruption attack introduced above to disrupt all benign client's contributions in privacy-preserving quantized FL. MUD-PQFed is able to identify malicious clients in coarse-grained means and fine-grained means. The former means that the server can localize the malicious behaviors in a specific subset/group of clients, while the latter ultimately identifies the exact malicious clients. The fine-grained detection is more preferred from the incentive perspective, as the exact malicious clients will be penalized, instead of a subset/group of clients that may also consist of benign clients.

\subsection{System Overview}

MUD-PQFed achieves the above goals by delicately adapting a state-of-the-art privacy-preserving data aggregation technique \cite{dekker2021privacy} for quantized FL.
By novel integration of such robust and secure aggregation technique with quantized FL, MUD-PQFed presents the first systematic solution for privacy-preserving quantized FL with robustness against malicious clients in the \textit{single-server} setting.
We firstly present an overview of the MUD-PQFed system that has four stages, followed by elaborations on the implementation details per stage.

\begin{enumerate}
    \item \textbf{Stage 1: Registration.} At this stage, each client sends a message to the server indicating that they want to contribute to the FL. After the server receives the messages sent by $n_c$ clients, the server divides these clients into groups according to the hypermesh.
    
    \item \textbf{Stage 2: Submission.} At this stage, each client generates a secret shared value $s$ for each group (i.e., the summation of the secret shares is 0) in which the client belongs to mask the gradient information to be $m$, and generates a commitment $d$ of the secret shared value $s$. The client sends the masked gradient information $m$ that is the ciphertext and commitment $d$ to the server.

    \item \textbf{Stage 3: Detection.} At this stage, the server performs three rounds of checks to identify and remove potentially malicious clients. The server first determines whether the summation of the secret shared value is 0 by calculating the commitment value of each group, then checks whether the value sent by the client to each group is equal, and finally checks whether the gradient sum of each group is within a reasonable range.
    
    \item \textbf{Stage 4: Aggregation.} At this stage, the server performs aggregation on parameters of clients after excluding detected malicious clients to update the global model, which will be sent to the clients for next round FL global model update till convergence or a predefined condition.
\end{enumerate}
\subsection{Stage 1: Registration}

The registration phase is to group clients and prepare parameters required by the remaining stages. First, each client generates his/her own public and private key pair ($pk,sk$), and sends the public key $pk$ to the server, indicating that he/she wants to participate in the FL. The server selects the $n_c$ clients to contribute to the aggregation and a commitment generator $\textsf{g}$. Once these parameters are received by the server/clients, and $n_c$ clients are firmed, the server groups these clients according to the following grouping method.

\textbf{Grouping Clients.} The groups are divided according to the hypermesh, as shown in Fig.~\ref{fig:hypermesh}. The server chooses a $d^n$ hypermesh, where $n_c=d^n$, $d \ge 2$ and $n \ge 2$. The edges that is either a row or a column in the hypermesh are the groups that a client is in. 
For example, the client $0$ in Fig.~\ref{fig:hypermesh} belongs to two groups of C\#0 and R\#0. Each client has a unique identifier, and two clients whose identifiers differ by one bit are neighbors to each other, e.g., (client 0 and client 4) or (client 0 and client 1). 
The $n$ groups in which client $i$ belongs is denoted as $\mathcal{G}_i$ that is a set---each client belongs to two groups in the exemplified Fig.~\ref{fig:hypermesh}. The group in which $d$ clients belong are denoted as $\mathcal{U}_j$ that is a set too. And the neighbors of client $i$ are denoted as $\mathcal{N}_i$. 

In the registration stage, once clients are grouped, the clients will generate the parameters required for the remaining stages.

\textbf{Generating parameters.} The secret sharing is a main and also popular method to realize MPC~\cite{erkin2012private}, which we leveraged herein to enable the MPC-based privacy preserving FL. In order to minimize the increased communication overhead when any two clients have to communicate with each other for sharing parameters by establishing a secure communication channel. The server forwards and then broadcasts the message. 
So that each client only needs to establish a secure communication channel with the server. In this context, each client firstly generates his/her own key pair
($pk,sk$), and sends the public key $pk$ to the server. 
Once the information from $n_c$ clients has been collected, the server starts grouping. Then server sends public key ${pk_q|q \in \mathcal{N}_i}$ to each client that is the neighbor of client $i$. 
To reduce communication overhead, the public key can be reused through all FL training rounds to avoid public key generation and broadcasting per round. In each round $t$, client $i$ needs to generate a random number $r_{i\to q,t}$---($i\to q,t$) means that client $i$ generates and sends this random number for client $q$ at round $t$---for each neighbor $q$ $\in$ $\mathcal{N}_i$, encrypts this random number with the public key $pk_q$ and sends it to the server. The client encrypts the random number to prevent the server from knowing the real value of the random number. 
After the client $i$ receives random numbers from all neighbors, the client $i$ calculates the secret sharing value
\begin{equation}
s_{i,j,t}=\sum_{q \in \mathcal{G}_j}{r_{i\to q,t}-r_{q\to i,t}}.
\end{equation}
For each $j$ $\in$ $\mathcal{G}_j$, with this secret sharing method construction we can get 
\begin{equation}
\forall{j \in \mathcal{G}:}\sum_{i \in \mathcal{U}_j}{s_{i,j,t}=0}.    
\end{equation}
That is, summation of secret shared values of each group is 0.

\subsection{Stage 2: Submission}

The submission phase requires the client to send parameters to the server. Client $i$ adds its secret shared value $s_{i,j,t}$ to its own gradient information ${\bf w}_{i,j,t}$ to get ciphertext ${\bf m}_{i,j,t}$. To prevent clients from sending different ${\bf w}_{i,j,t}$ to different groups in the same round, the client must also make a commitment on its secret shared value. We use a simple computationally bound and computationally hidden commitment scheme. That is, the commitment value does not reveal any information about the message ${\bf m}$ and the receiver can determine that $\bf m$ fits the information corresponding to the commitment. 
To submit a ciphertext $\bf m$, a client needs to send $\textsf{g}^m$. So each client needs to calculate commitment ${\bf d}_{i,j,t}=\textsf{g}^{s_{i,j,t}}$ and then sends $\{({\bf m}_{i,j,t},{\bf d}_{i,j,t}) | j \in \mathcal{G}_i\}$ to the server.

\subsection{Stage 3: Detection}

The detection phase is to detect malicious clients and remove them. There are three rounds of detection. 

\textbf{First Round.} The server checks whether the sum of the secret shared values of all clients in each group is 0 by computing 
\begin{equation}
\label{eau4}
\prod_{i \in\mathcal{U}_j}{{\bf d}_{i,j,t}}= \prod_{i \in \mathcal{U}_j}{\textsf{g}^{s_{i,j,t}}}=\textsf{g}^{\sum_{i \in \mathcal{U}_j^{s_{i,j,t}}}}=\textsf{g}^0=1.  
\end{equation}
If the final result of a group $j$ is not 1, it means that at least one malicious client in group $j$ submitted the wrong secret shared value, so $j$ will be added to $\mathcal{V}$ that is a set used to record the groups of potential malicious clients.

\textbf{Second Round.} For each client, the server computes

\begin{equation}
\begin{aligned}
\label{equ5}
\{\textsf{g}^{m_{i,j,t}}({\bf d}_{i,j,t})^{-1}|j\in \mathcal{G}_i\}&=\{\textsf{g}^{w_{i,t}+s_{i,j,t}}\textsf{g}^{-s_{i,j,t}}|j\in \mathcal{G}_i\}\\
&=\{\textsf{g}^{w_{i,t}}|j\in \mathcal{G}_i\}
\end{aligned}
\end{equation}

to verify that the ${\bf w}_{i,t}$ sent by the client $i$ to each group is the same. If client $i$ fails the detection, all groups in $\mathcal{G}_i$ will be added to $\mathcal{V}$.

\textbf{Third Round.} For group $j$, the server computes
\begin{equation}
\tilde{\bf w}_{j,t}=\sum_{i\in \mathcal{U}_j}{{\bf c}_{i,j,t}}=\sum_{i\in \mathcal{U}_j}{({\bf w}_{i,t}+{\bf s}_{i,j,t})}=\sum_{i\in \mathcal{U}_j}{{\bf w}_{i,t}}.
\end{equation}
If the aggregate value $\tilde{\bf w}_{j,t}$ is not within a reasonable range, it means that at least one client in $j$ has submitted malicious gradient information, so all clients in $j$ will be added to $\mathcal{V}$.

We can identify malicious clients by checking which clients appear $n$ times in $\mathcal{V}$. 
This is because it is divided according to the $d^n$ hypermesh, each client will belong to $n$ groups at the same time, so the malicious client will belong to the $n$ groups at the same time. During detection, all the groups that the malicious client belongs to are added to $\mathcal{V}$. Therefore, malicious client(s) can be determined by counting which client(s) in $\mathcal{V}$ appear(s) in all $n$ groups \textit{at the same time}.

\subsection{Stage 4: Aggregation}

After detecting a malicious client, it is required to remove the submitted values of all groups where the malicious client belongs. Because as long as there is a malicious client in group $j$, the submitted value of this group $j$ is invalid, so the server has to remove the submitted value of the entire group $j$. After removal, the model parameter sum of each group is 
\begin{equation}
\tilde{\bf w}_t=\sum_{j\in \mathcal{G}\setminus \mathcal{V}}{\tilde{\bf w}_{j,t}}.
\end{equation}
Then the server calculates the number of valid clients $q$ according to the number $p$ of valid groups, where $q=p\times d$ and uses
\begin{equation}
\bar{\bf w_t}=\frac{\tilde{\bf w}_t}{q}    
\end{equation}
to update the global model parameter.

\section{Experiments}\label{sec:experiments}

\subsection{Setup}

Once our protocol selects $n_c$ clients to participate in the registration phase, the clients participating in the protocol will not be changed in all subsequent rounds, so our protocol satisfies cross-silo FL~\cite{zhang2020batchcrypt,kairouz2021advances}.

\vspace{0.2cm}
\noindent{\bf Dataset.} To validate the efficacy and effectiveness of our MUD-PQFed, we use two commonly used datasets, MNIST~\cite{lecun1998gradient} and CIFAR-10~\cite{chang2020multi}. Both datasets have been used in recent related work~\cite{dekker2021privacy,mills2019communication}. In all experiments, we divide the dataset equally according to the IID distribution following~\cite{so2020byzantine,mo2021ppfl}.

The MNIST dataset consists of handwritten digital images. There are 10 categories, corresponding to 10 Arabic numbers from `0' to `9'. The numbers of training and testing image samples are 60,000 and 10,000, respectively. Each gray image sample has a size of $28\times 28 \times 1$.

The CIFAR-10 dataset is a small-scaled dataset used to identify pervasive objects. There are 10 categories of RGB color images: airplane, automobile, bird, cat, deer, dog, frog, horse, ship and truck. The size of each image is $32\times 32 \times 3$, and there are 6000 images per category. In total, there are 50,000 training images and 10,000 testing images.

\vspace{0.2cm}
\noindent{\bf Environment.} The MUD-PQFed is implemented in Python. We run the experiments on a server with an AMD Ryzen 5 4600H CPU, 16GB RAM with the Windows 10 operating system. 

As for the used model architecture, we follow the same neural network architecture which are customized for ternary FL and binary FL neural networks and parameter settings as ~\cite{xu2020ternary} and ~\cite{yang2021communication}, respectively\footnote{The quantization can be applied layer-by-layer~\cite{xu2020ternary,yang2021communication}}.

We first evaluate the MUD-PQFed efficacy that is the accuracy when the model corruption attack is performed. We then evaluate the effectiveness that is the detection capability of the MUD-PQFed.
\subsection{Accuracy Evaluation}

All experiments in this part use $d^2$-hypermesh for grouping, while other $d^n$-hypermesh evaluations are deferred to Section~\ref{sec:DetCap}. To be precise, if the number $n_c$ of clients participating in the aggregation is 25, then the server divides the 25 clients according to the $5^2$-hypermesh. That is, each group has 5 clients, each client belongs to two groups at the same time. There are a total of 10 groups. 

\subsubsection{Ternary FL}

\textbf{Baseline:} The ternary FL by Xu \textit{et al.}~\cite{xu2020ternary} is used as a baseline. There are no malicious clients in this baseline, where all clients and the server honestly participate in the aggregation. The accuracy of the ternary and binary FL after convergence is about 92.75\% and 90.1\%, respectively.

\vspace{0.2cm}
\noindent{\bf MNIST.} For the experiments of MNIST in ternary FL, model architecture is a multiple layer perception (MLP) network. The number of parameters that each client needs to submit in each round of aggregation is 24,330. We assume that a malicious client has tampered the parameters in one of the layers consisting of 600 parameters. Each parameter is incremented by a random number sampled within [20, 30].
For instance, the original real parameter value is 1, and the malicious client will change it to any value of [21,31]. A total of 30 rounds FL training/aggregation are performed in our experiments. 
We assume that malicious clients start tampering from round $6_{\rm th}$ for every 5 rounds. 
To wit, the malicious client does not change parameter values in rest rounds. Fig.~\ref{fig:termnist1} (a) with \textit{one} malicious client and Fig.~\ref{fig:termnist1} (b)  with \textit{two} malicious clients show the relationship between the number of clients and accuracy for the MNIST dataset of ternary FL.
It can be seen that after the $6_{\rm th}$ round model corruption attack, if the malicious clients are not detected and removed, the accuracy after aggregation will decrease significantly.
In most cases, the accuracy cannot be recovered later, completely disrupts the usage of the global model and discourages benign clients contribution incentives. 
The fewer the number of clients---$n_a$ to $n_c$ proportion is higher, the more severe the attack. 
This is under expectation, when the ratio of $n_a$ to $n_c$ is small, the impact of a single client's adverse contribution to the total aggregated results will be less.
 
Once our MUD-PQFed scheme is applied, there is no drop in accuracy. Since the MUD-PQFed will remove the submitted values of malicious clients before FL aggregation after detecting malicious clients. So that the model corruption attack is prevented. 

The number of malicious clients (i.e., Fig.~\ref{fig:termnist1} with one and two malicious client(s), respectively) has different effects on the accuracy. 
The higher the number of malicious clients in the aggregation process, the worse the accuracy will be without MUD-PQFed---one round attack at the $6_{\rm th}$ can always successfully corrupt the model as in Fig.~\ref{fig:termnist1} (b) 
Because the greater the number of malicious clients, the more parameters will be tampered and received during the aggregation process. 
Nonetheless, once the MUD-PQFed is applied, the model corruption attack will be prevented, as the model accuracy witnesses no drop.

When there are two malicious clients, we have compared two different cases: both in the same group and two different groups.
Fig.~\ref{fig:termnist2} depicts the accuracy with/without applying MUD-PQFed.

From Fig.~\ref{fig:termnist2} (a), it can be clearly observed that model corruption attack succeeds regardless of two clients being in the same group or different groups---no notable differences in these two cases.

Because two malicious clients in any case tamper the same number of parameters, thus the global model damage tends to be same.

The results in Fig.~\ref{fig:termnist2} (b) demonstrate that our MUD-PQFed works well for both cases, as the accuracy drop caused by the malicious clients manipulated parameters has all been successfully recovered after removing these malicious clients. Next, we use an example to explain the rationale.

Suppose $n_c$=16, that is, the number of clients participating in the aggregation is 16, denoted as (0, 1, ... , 15). There will be eight groups in total, which are $G_0=[0,1,2,3]$, $G_1=[4,5,6,7]$, $G_2=[8,9,10,11]$, $G_3=[12,13,14,15]$, $G_4=[0,4,8,12]$, $G_5=[1,5,9,13]$, $G_6=[2,6,10,14]$, $G_7=[3,7,11,15]$. There are 32 copies and we call each copy a submitted value to ease description.
Suppose two malicious clients 0 and 1 are in the same group, the groups that need to be removed during aggregation will be $G_0$, $G_4$ and $G_5$ as they contain clients 0 and 1.
A total of 12 submitted values need to be deleted. Suppose two malicious clients 0 and 5 are not in the same group, the groups that need to be removed when aggregated are $G_0$, $G_1$, $G_4$ and $G_5$. A total of 16 submitted values need to be deleted.
Although two malicious clients have deleted 4 more submitted values when they are in different groups, there are still 16 remaining submitted values, which retains to be same as the number of submitted values of the baseline. 
It can be understood that the number of submitted values is 'saturated'. In other words, after removing 3 groups of submitted values, 20 commits are 'over saturated' when two malicious clients are in the same group. Therefore, our solution achieves the best accuracy in both cases of malicious clients in the same group and different groups, and there is no global model accuracy gap after malicious clients removal.

Fig.~\ref{fig:termnistbaseline} shows the accuracy of the baseline when there are no malicious clients and the accuracy when the number of MUD-PQFed malicious clients is 1. It can be seen that the accuracy of MUD-PQFed is basically the same as the baseline accuracy, which affirms the efficacy of our scheme.
 
It is also explained that whether the malicious clients are in the same group or not has no effect on the accuracy of our scheme. 

\vspace{0.2cm}
\noindent{\bf CIFAR-10.}
For the experiments of CIFAR-10 in ternary FL, the client uses a customized CNN with a number of parameters of 1,166,176 for training.

We assume that a malicious client has changed the parameters of one of the layers which has 18,432 parameters, and each parameter is incremented by a random number sampled in [20, 30]. Fig. \ref{fig:tercifar1} (a) depicts the accuracy when there is a single malicious client who can greatly degrade the global model accuracy without applying MUD-PQFed. 

Once the MUD-PQFed is applied, it successfully mitigates the accuracy drop.

Fig.~\ref{fig:tercifar1} (b) shows the accuracy when two malicious clients exist. It can be observed that the tendency is similar to the accuracy tendency of the case when there is only one malicious client. 
One malicious client changes 18,432 parameters, which is only 1.6\% out of the total number of 1,166,176 parameters, so when the number of malicious clients is 2, the accuracy does not drop to the point of being irrecoverable like the MNIST dataset. 
Nonetheless, the application of MUD-PQFed effectively eliminates the attack effects. Same as the MNIST, Fig.~\ref{fig:tercifar2} depicts the accuracy with/without applying MUD-PQFed for CIFAR-10. Similarly, we can see that if malicious clients are not removed, once malicious clients launch attack, the accuracy will drop. Using  MUD-PQFed can perfectly eliminate the influence of malicious clients.

\begin{figure}[t!]
    \centering
    \includegraphics[width=0.5 \textwidth]{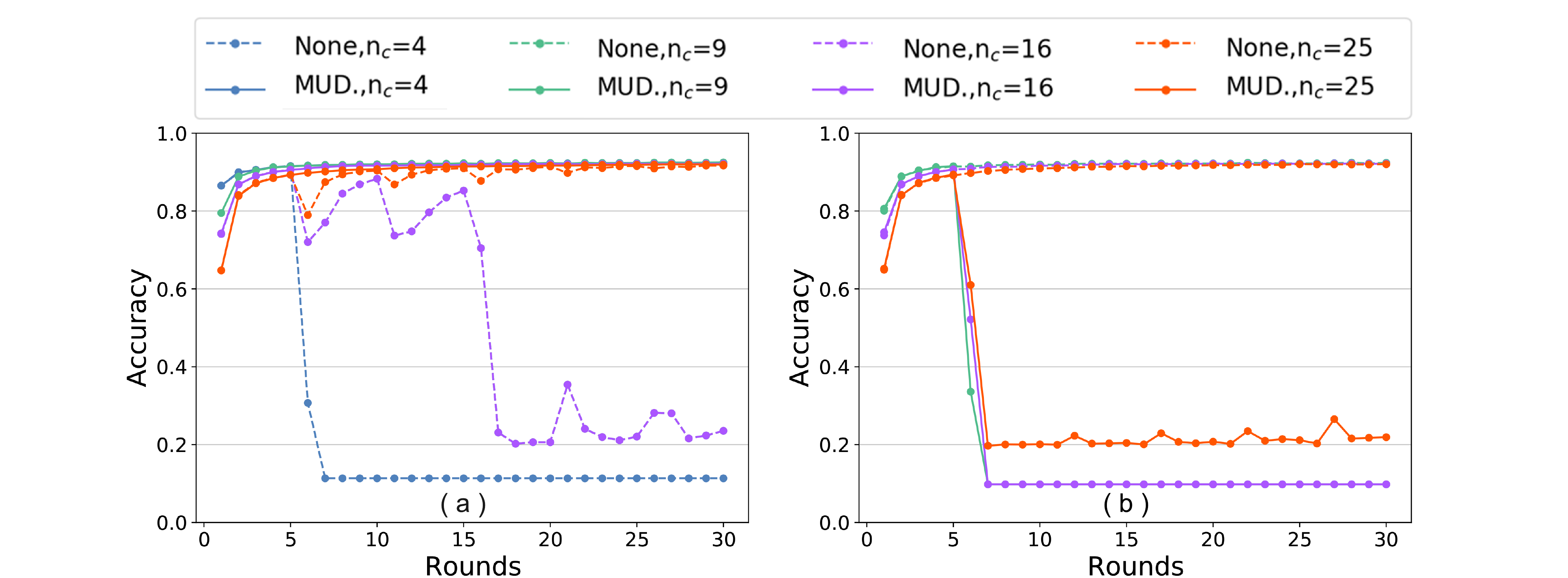}
    \caption{The relationship between accuracy and the number of malicious clients ((a) one malicious client and (b) two malicious clients) in ternary FL (MNIST). Note that ‘none’ means that the MUD-PQFed is not applied when there exists malicious clients. }
    \label{fig:termnist1}
\end{figure}

\begin{figure}[t!]
    \centering
    \includegraphics[width=0.5\textwidth]{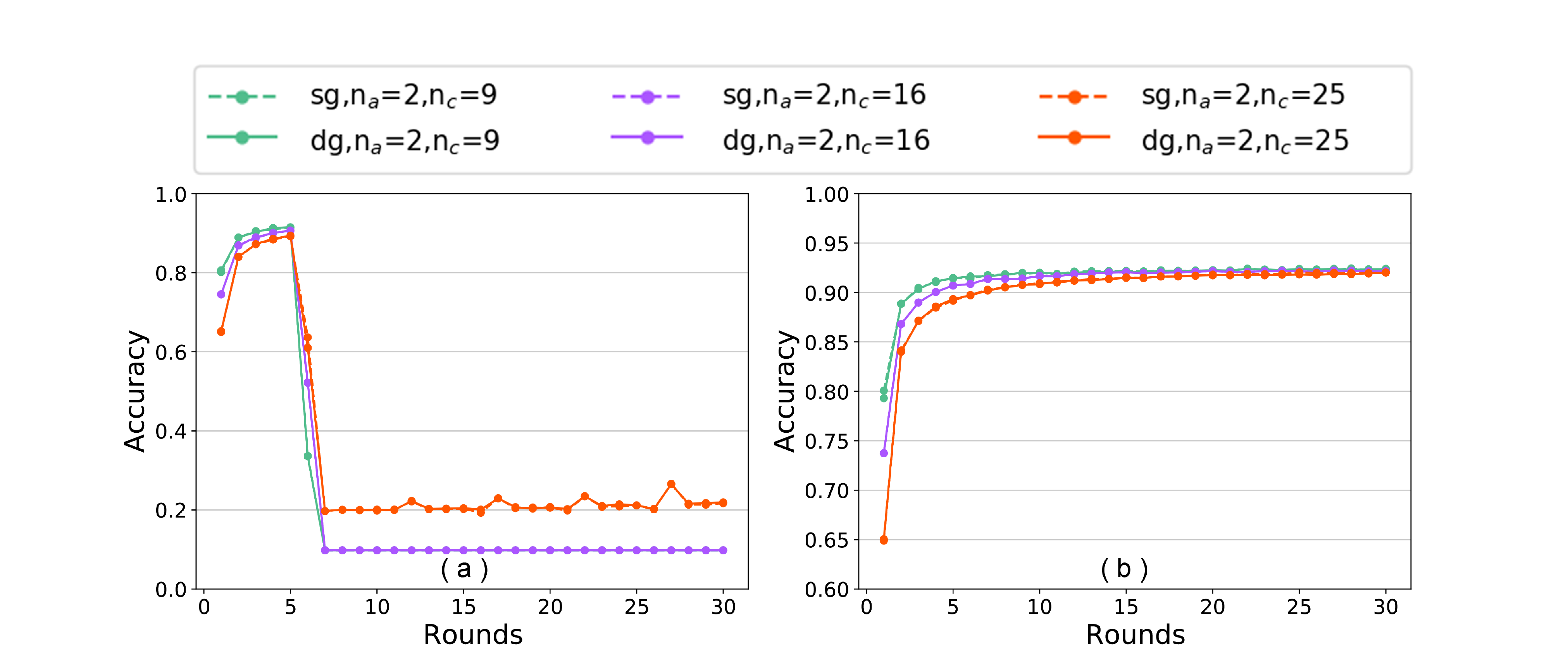}
    \caption{Accuracy  (a) without and  (b) with applying
MUD-PQFed in ternary FL (MNIST). Note that 'sg' means that two malicious clients in the same group and 'dg' means two malicious clients in the different group. }
    \label{fig:termnist2}
\end{figure}

\begin{figure}[t!]
    \centering
    \includegraphics[width=0.45\textwidth]{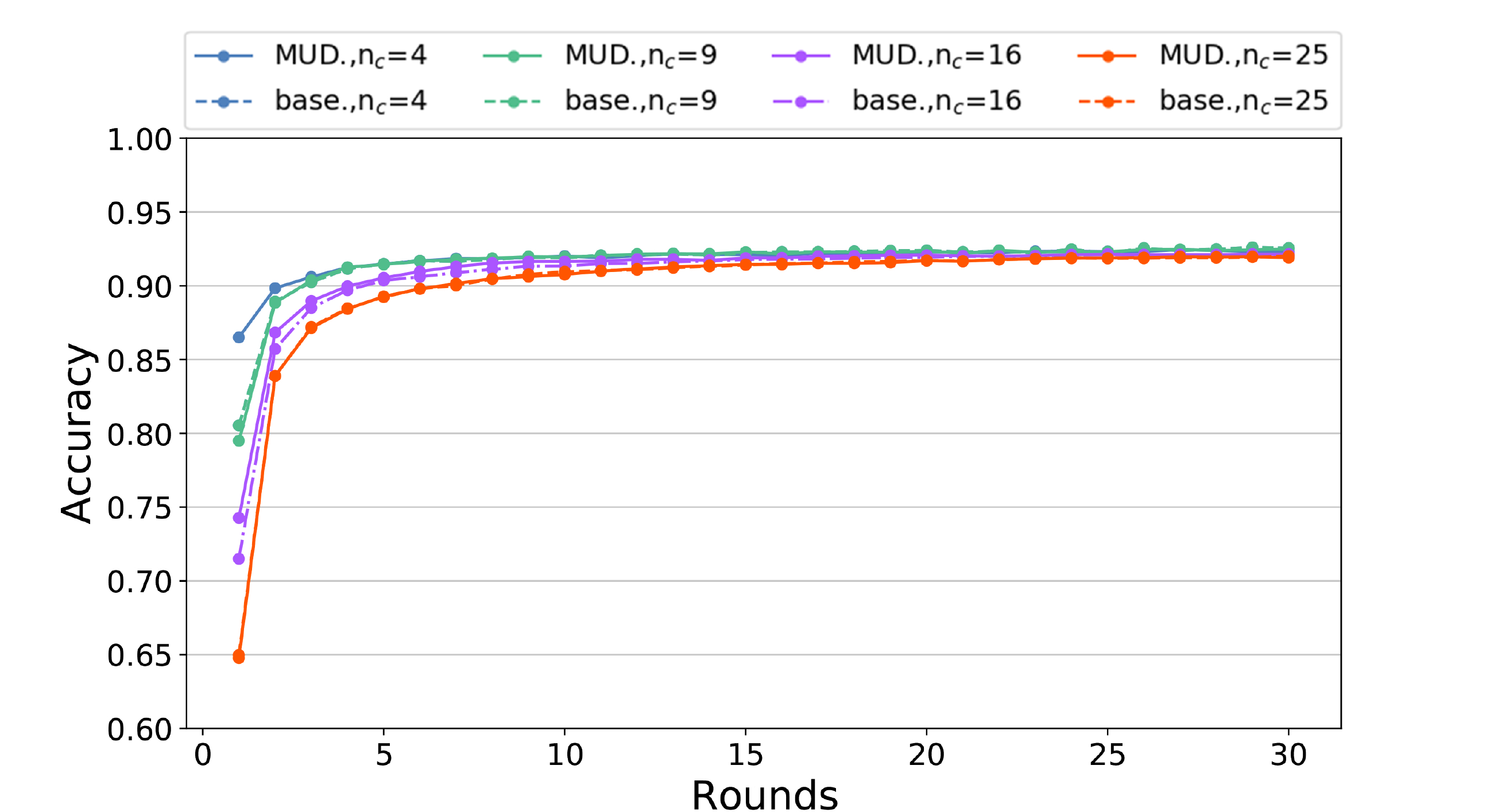}
    \caption{\centering Accuracy comparison of MUD-PQFed with baseline (MNIST).}
    \label{fig:termnistbaseline}
\end{figure}

\begin{figure}[t!]
    \centering
    \includegraphics[width=0.5\textwidth]{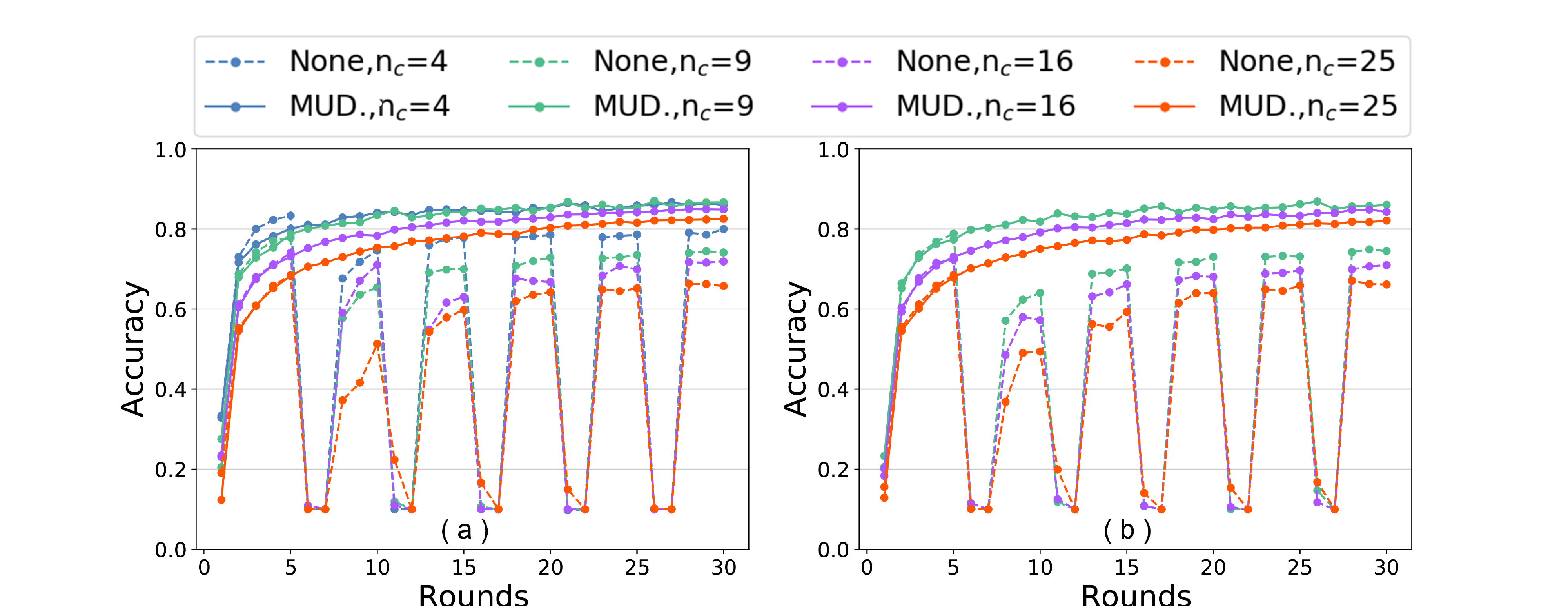}
    \caption{The relationship between accuracy and the number of malicious clients(one malicious client (a) and two malicious clients (b)) in ternary FL (CIFAR-10). }
    \label{fig:tercifar1}
\end{figure}

\begin{figure}[t!]
    \centering
    \includegraphics[width=0.5\textwidth]{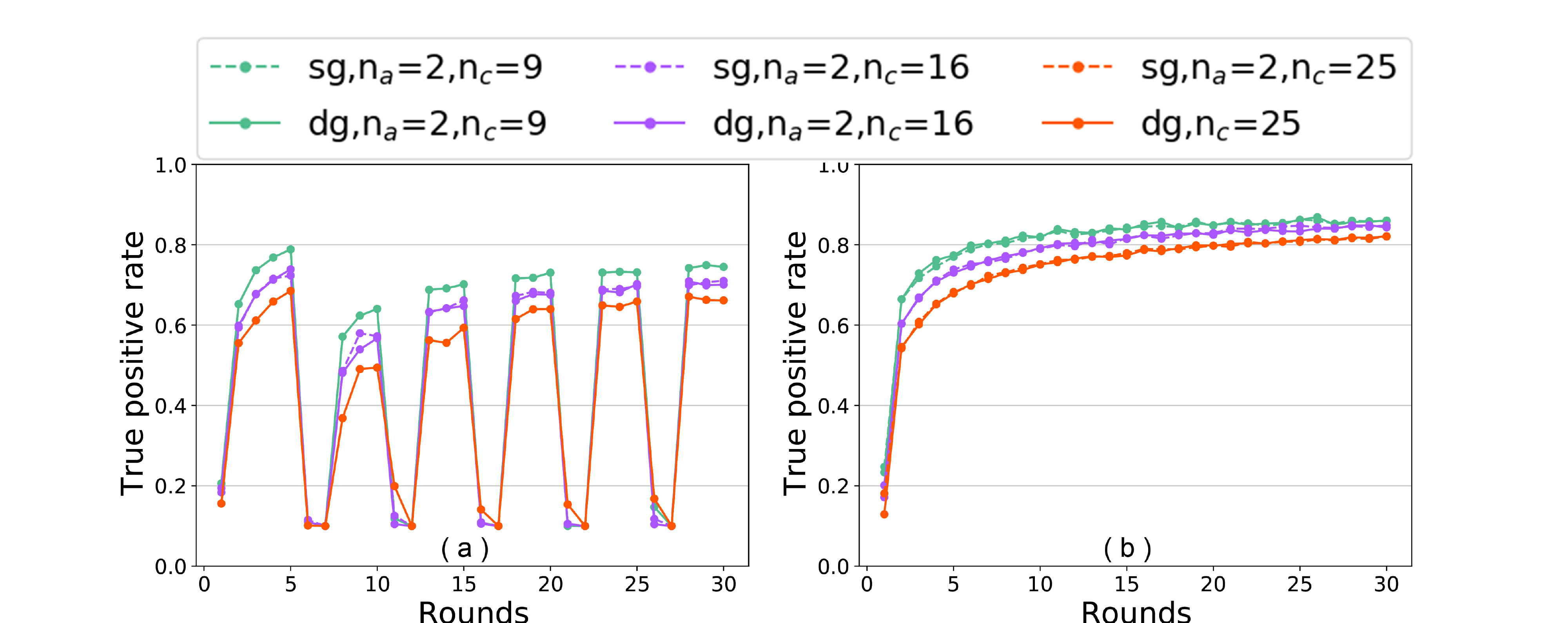}
    \caption{Accuracy  without (a) and  with (b) applying
MUD-PQFed in ternary FL (CIFAR-10). }
    \label{fig:tercifar2}
\end{figure}

\subsubsection{Binary FL}

\textbf{Baseline:} The binary FL by Yang \textit{et al.}~\cite{yang2021communication} is used as a baseline. Similar to ternary FL, the baseline accuracy is obtained in the absence of malicious clients.

We conducted experiments on the MNIST dataset. The number of parameters that each client needs to submit in each round of aggregation is 82,242 by following the same model architecture~\cite{yang2021communication} (i.e., a binarized neural network customized for the MNIST dataset). We assume that a malicious client changes parameters of one of the layers with 1,000 parameters: each parameter
is incremented by a random number sampled in [20, 30].
Fig. \ref{fig:binmnist} (a) shows the accuracy when there is only one malicious client. Similar to ternary FL, the malicious client can greatly degrade/corrupt the accuracy. Fig.~\ref{fig:binmnist} (b) shows the accuracy when two malicious clients exist. 
 
The Fig.~\ref{fig:binmnistbaseline} compares the accuracy after applying the MUD-PQFed with the baseline. Results in Fig.~\ref{fig:binmnist} and Fig.~\ref{fig:binmnistbaseline} further affirm the efficacy of the MUD-PQFed of mitigating the malicious clients attack effect when for the binary FL.

\begin{figure}[t!]
    \centering
    \includegraphics[width=0.5\textwidth]{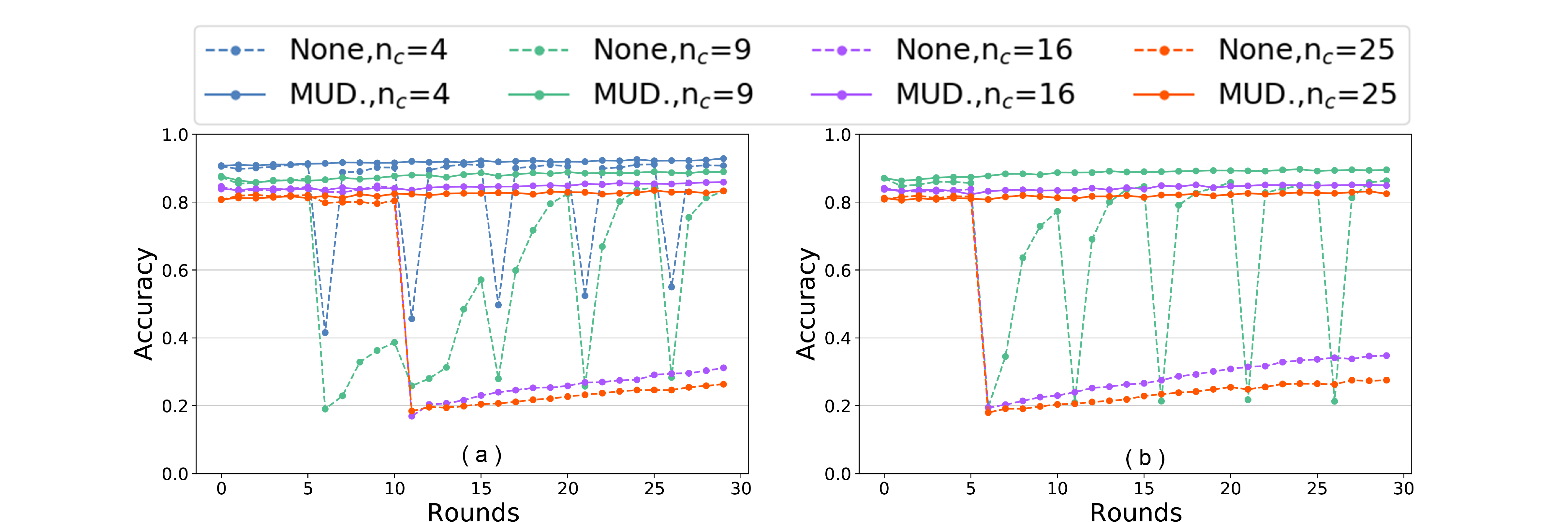}
    \caption{The relationship between accuracy and the number of malicious clients(one malicious client (a) and two malicious clients (b)) in binary FL (MNIST).}
    \label{fig:binmnist}
\end{figure}

\begin{figure}[t!]
    \centering
    \includegraphics[width=0.45\textwidth]{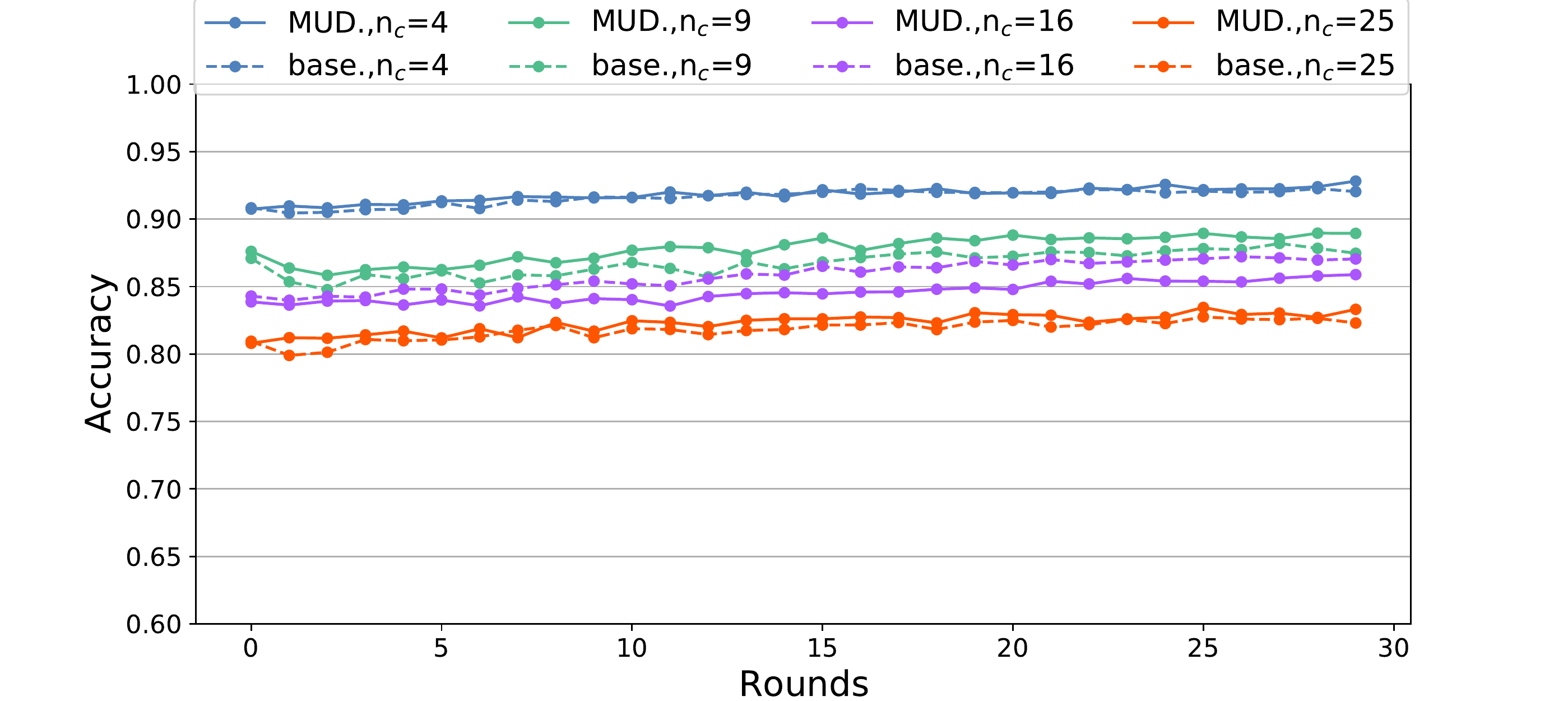}
    \caption{\centering Accuracy comparison of MUD-PQFed with baseline (MNIST).}
    \label{fig:binmnistbaseline}
\end{figure}

\subsection{Detection Capability}\label{sec:DetCap}

The above experiments validate the efficacy of the MUD-PQFed to retain the global model accuracy under existing of malicious clients performing model corruption attacks in comparison with the baseline.
Now we evaluate its effectiveness in terms of detection capability from true positive rate and false positive rate.

\noindent\textbf{True positive rate (TPR).} The TPR is the probability of malicious clients being correctly identified.

\noindent\textbf{False positive rate (FPR).} The FPR is the probability of a benign client being misidentified as a malicious client.

Before diving into experimental evaluations of the detection capability, we first give a theoretical analysis of the detection capability.

\begin{theorem}
\label{thm}
\textbf{(Detection Capability)} When $n_c$ clients are grouped by the $d^n$-hypermesh, the MUD-PQFed protocol can i) always provide a 100\% TPR, and ii) will not misidentify benign clients as malicious clients (thus 0\% FPR) conditioned on that the number of malicious clients is $\le$ $n$ and the \textit{equal sign} holds when at least two malicious clients are in the \textit{same group}. 
\label{th1}
\end{theorem}

\begin{proof}
A group only contains clients with a difference of one(-bit), two clients can belong to at most one group. If a benign client is identified as a malicious client, then this client will be in $n$ groups in $\mathcal{V}$ at the same time, and there will be one malicious client in each group. So there are $n$ malicious clients in total and they are in different groups. Since we assume that the number of malicious clients is $\le$ $n$ and the equal sign holds when at least two malicious clients are in the same group, therefore, as long as the number of malicious clients is $\le$ $n$ and the equal sign holds when at least two malicious clients are in the same group, our scheme will not misidentify normal clients, guaranteeing a 0\% FPR.
\end{proof}

\begin{theorem}
It is impossible for malicious client $i$ to send a message to the aggregator in round $t$ such that ${\bf w}_{i,j,t}$ $\neq$ ${\bf w}_{i,j^{'},t}$ for any two groups $j,j^{'}\in \mathcal{G}_i$ such that the verification of the aggregator doesn't fail, assuming that the discrete logarithm problem is intractable in the group generated by $\textsf{g}$.
\end{theorem}
\begin{proof}
First, if client $i$ or any neighbor $k$ $\in$ $\mathcal{N}_i$ fails to send a message, malicious client will be immediately detected by the aggregator. 
Second, the aggregator will follow Equation~\ref{eau4} to verify $\sum_{i \in \mathcal{U}_j}{s_{i,j,t}=0}$. Finally, from Equation \ref{equ5}, it can be seen that for fixed $t$ and $i$ $\in$ $\mathcal{U}$, all  ${\bf m}_{i,j,t}$- ${\bf s}_{i,j,t}$ for $j \in \mathcal{G}_j$ are equal. Therefore, by
definition of ${\bf w}_{i,j,t}$, all ${\bf w}_{i,j,t}$ for fixed $i \in \mathcal{U} $ and $t$ are also equal.
\end{proof}

\begin{table}[t!]
    \centering
    \caption{Detection Capability of $d^2$-hypermesh (MNIST)}
    \begin{tabular}{|c|c|c|c|c|c|}
    \hline
    \tabincell{c}{$n_c$}&\tabincell{c}{$d$} & \tabincell{c}{$n_a$} & \tabincell{c}{Same \\group} & \tabincell{c}{TPR} & \tabincell{c}{FPR} \\ 
      \hline
      4&2&1& \Checkmark &100\%&0.00\% \\ 
      \hline
     9&3&1&\Checkmark&100\%&0.00\% \\ 
     \hline
     9&3&2&\Checkmark&100\%&0.00\% \\ 
     \hline
     9&3&2&\XSolidBrush&100\%&28.60\% \\ 
     \hline
      16&4&1&\Checkmark&100\%&0.00\% \\ 
      \hline
      16&4&2&\Checkmark&100\%&0.00\% \\ 
      \hline
      16&4&2&\XSolidBrush&100\%&14.29\% \\ 
      \hline
     25&5&1&\Checkmark&100\%&0.00\% \\ 
     \hline
     25&5&2&\Checkmark&100\%&0.00\% \\ 
     \hline
     25&5&2&\XSolidBrush&100\%&8.69\% \\ 
      \hline
     64&8&1&\Checkmark&100\%&0.00\% \\ 
     \hline
      64&8&2&\Checkmark&100\%&0.00\% \\ 
     \hline
     64&8&2&\XSolidBrush&100\%&3.23\% \\ 
      \hline
     
    \end{tabular}
    
    \label{tab:Recognition Rate1}
\end{table}

\begin{table}[t!]
    \centering
    \caption{Detection Capability of $d^3$-hypermesh (MNIST)}
    \begin{tabular}{|c|c|c|c|c|c|}
    \hline
    \tabincell{c}{$n_c$}&\tabincell{c}{$d$} & \tabincell{c}{$n_a$} & \tabincell{c}{Same\\ group} & \tabincell{c}{TPR} & \tabincell{c}{FPR} \\  
      \hline
      8&2&1&\Checkmark&100\%&0\% \\ 
      \hline
      8&2&2&\Checkmark&100\%&0\% \\ 
      \hline
      8&2&2&\XSolidBrush&100\%&0\% \\ 
      \hline
      8&2&3&\Checkmark&100\%&0\% \\ 
      \hline
      8&2&3&\XSolidBrush&100\%&20\% \\ 
      \hline
     27&3&1&\Checkmark&100\%&0\% \\ 
      \hline
     27&3&2&\Checkmark&100\%&0\% \\ 
      \hline
     27&3&2&\XSolidBrush&100\%&0\% \\ 
     \hline
     27&3&3&\Checkmark&100\%&0\% \\ 
     \hline
     27&3&3&\XSolidBrush&100\%&0\%\\ 
     \hline
     64&4&1&\Checkmark&100\%&0\% \\ 
      \hline
     64&4&2&\Checkmark&100\%&0\% \\ 
      \hline
     64&4&2&\XSolidBrush&100\%&0\% \\ 
     \hline
     64&4&3&\Checkmark&100\%&0\% \\ 
     \hline
     64&4&3&\XSolidBrush&100\%&0\% \\ 
     \hline
    \end{tabular}
    
    \label{tab:Recognition Rate2}
\end{table}

\vspace{0.2cm}
\noindent{\bf $d^2$-hypermesh.} Table~\ref{tab:Recognition Rate1} details experimental results when the clients are divided according to the $d^2$ hypermesh with varying $d$. The MNIST is trained with the ternary FL. Same to previous setup, a malicious client changes 600 out of of the 24,330 total parameters by adding a sampled value in [20, 30] per tampered parameter. 
The $4_{\rm th}$ column indicates whether the malicious clients are in the same group. '\Checkmark' means that they belong to the same group; otherwise, '\XSolidBrush'.

The FPR is detailed in the $6_{\rm th}$ column.
 
For example, if $n_c$=16 and $n_a$=2---number of benign clients is 14, when two benign clients are misidentified as malicious clients, the FPR is 2/14=0.1429. 
From Table~\ref{tab:Recognition Rate1}, we can see that the increase of malicious clients $n_a$ does not affect TPR, which affirms that the MUD-PQFed can always correctly identify malicious clients regardless of malicious clients being in the same/different group. But whether malicious clients are in the same group will affect the FPR. When the $n_a=2$ that does not meet $n_a<n=2$, the FPR occurs when these malicious clients are in \textit{different} groups. 

\vspace{0.2cm}
\noindent{\bf $d^3$-hypermesh.} Table~\ref{tab:Recognition Rate2} details the detection capability when the $d^3$-hypermesh is applied, where other experimental settings are same to the above $d^2$-hypermesh. The major difference here is that the $n$ is set to be 3, therefore, the FPR will only occur when the number of clients is no less than 3 compared to 2 in the $d^2$ hypermesh. This means that to reduce the FPR, a larger $n$ is preferred to avoid falsely punishment for benign clients, while the TPR can always be guaranteed to be 100\% regardless $n$ to penalize malicious clients.

\begin{figure}[t!]
    \centering
    \includegraphics[width=0.35\textwidth]{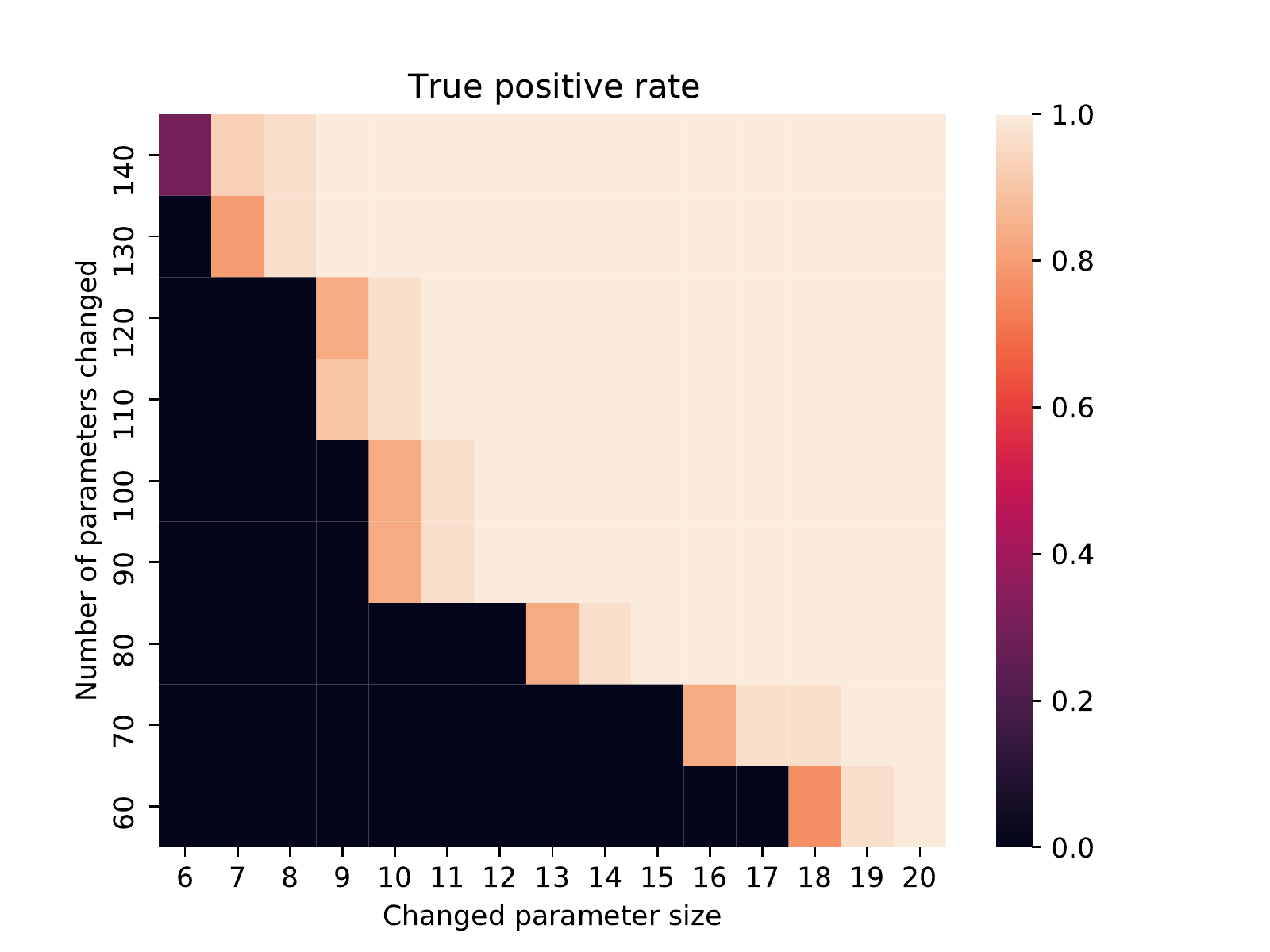}
    \caption{\centering True positive rate under different manipulation degree/size per parameter and number/fraction of parameter tampered (MNIST).}
    \label{fig:heatmap}
\end{figure}

\begin{figure}[t!]
    \centering
    \includegraphics[width=0.5\textwidth]{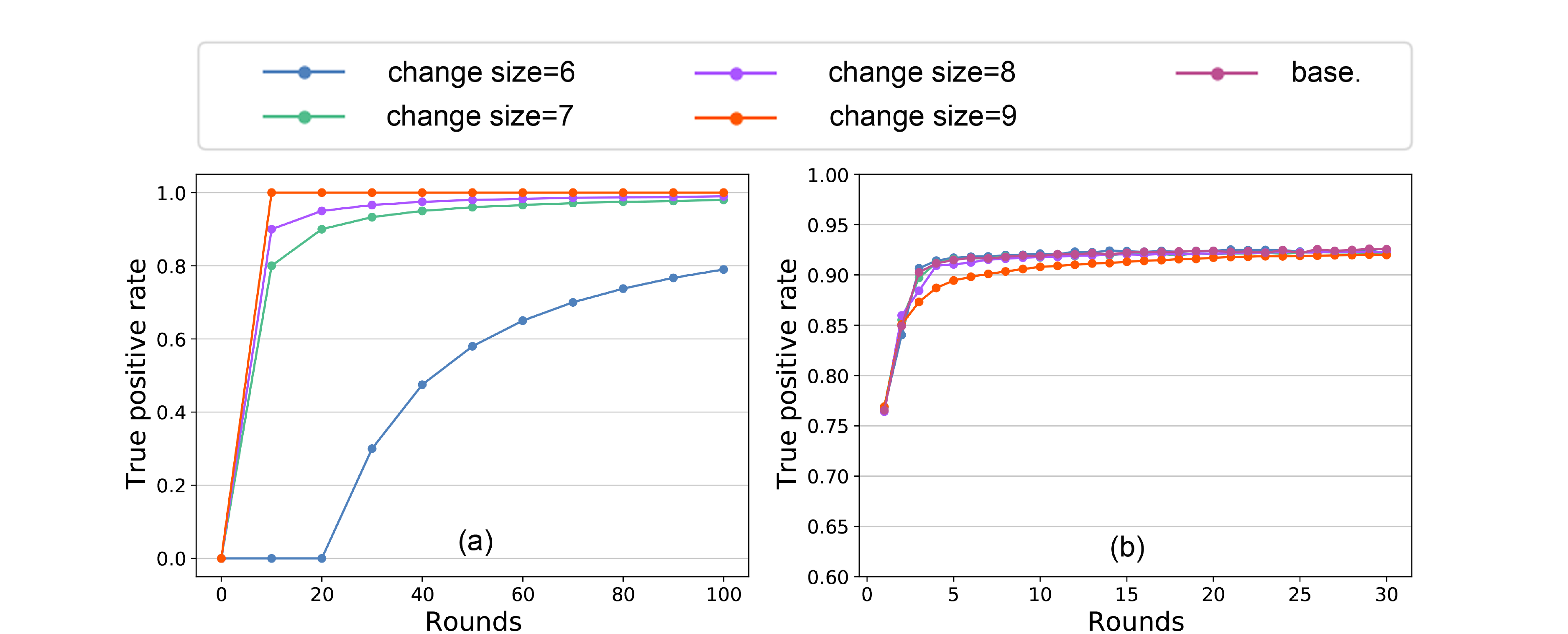}
    \caption{ True positive rate under different manipulation degree/size per parameter and number/fraction of parameter tampered (a).  Accuracy when the parameter change degree is small (b). (Ternary FL + MNIST).}
    \label{fig:TPR}
\end{figure}

\vspace{0.2cm}
\noindent{\bf Manipulation Degree on TPR.}
Fig. \ref{fig:heatmap} shows the TPR as a function of manipulation degree in terms of i) the number/fraction of parameters tampered and manipulation degree per parameter. The $n_c$ and $n_a$ are set to be four and one, respectively. 
The horizontal axis represents the degree per parameter changed, and the vertical axis represents the number of parameters tampered. 
It can be seen that with the increase in parameter degree change or/and the number of parameters changed, the TPR also increases whenever the attacker wants to successfully and efficiently perform the model corruption attack. 
The TPR is low due to limited number of tampered parameters or/and small degree change per parameter, as a result of the malicious tampered parameters can be concealed by the rest of values submitted from benign clients.
 
But note, at the same time, the attack impact on the global model accuracy is very small or even has no impact in our experiments, which can be trivially recovered by running a few more FL rounds once the attack disappears. That is, malicious clients cannot achieve the desired attack effect in this case. 

\vspace{0.2cm}
\noindent{\bf Attack Frequency on TPR.} Fig.~\ref{fig:TPR} depicts the relationship between the TPR and the attack frequencies (i.e., the number of rounds the malicious client submit manipulated parameters).
Here the number of parameters changed in set to be 140, while the degree per parameter is small, ranging from 6 to 9. There are two observations, larger manipulation degree, faster the attacker to be exposed in less number of rounds (i.e., about 10 rounds).
If the attacker wants to be stealthier (i.e., being detected slower), it has to lower the manipulation degree, which however could not succeed to corrupt the model accuracy. We validate the global accuracy corresponding to Fig.~\ref{fig:TPR} (a) in Fig.~\ref{fig:TPR} (b). It can be seen that when \textit{change size=6}, although our scheme cannot detect malicious clients in the first 20 rounds, the impact of malicious clients on the accuracy is almost none at this time. At the same time, we are able to detect malicious clients even though the attack has no accuracy impact after 20 rounds.

\section{Discussion}\label{sec:discussion}

\subsection{Group Contribution Deletion}
Though the MUD-PQFed, as the first work to our knowledge, considers the model corruption attack against quantized FL via illegitimate parameter submission and can fine-grained detect these attackers, it has to delete the parameters belong to the detected malicious client's group(s) to remove the attack effect in the aggregation stage. An ideal case should only delete the submitted value of the malicious client instead of the malicious client belonged group's parameters.

Notably, deleting the parameters in the malicious client belonged groups has no or merely negligible effect on the global model accuracy compared to the baseline, as we have extensively affirmed. This is because the FL itself can tolerate client leave, e.g., partial clients participation caused by `straggler' in the FL can still guarantee the model convergence to a satisfactory accuracy~\cite{gao2021evaluation}.

\subsection{Incorporating Incentive}
The main advantage of MUD-PQFed is to accurately identify malicious clients, and then the incentive mechanism can be incorporated to punish malicious clients in the FL. Note that FL relies on multiple local clients to jointly train the global model, the quality of client (i.e., data quality) will severely affect the global model utility. Nonetheless, clients, especially those high quality clients, are loathed to participate and share their own updates without adequate compensation and rewards. 

Firstly, clients participating in FL may have privacy concerns because of privacy inference attacks even if the local data is not accessed by the server and other clients~\cite{zhou2022ppa}.

Secondly, there will be system cost incurred to the client during participation in the FL. More specifically, it will inevitably consume the computational and communicational resources of the client.

Incentives can be mainly divided into positive incentives ~\cite{kang2019incentive,zhan2020learning} and reverse penalties ~\cite{gao2021fifl,bao2019flchain}. Positive incentives motivate clients by rewarding them, while negative incentives avoid malicious behavior by punishing individuals. So we can use the reverse incentive mechanism to punish the identified malicious clients, while rewarding the normal clients with the positive incentive mechanism 
If it only needs to eliminate the influence of malicious client submitted values on the final aggregation result, coarse-grained detection can be applied by simply grouping the clients participating in the aggregation and then check whether their aggregated values in a normal range per group.
Then the submitted parameters of the entire group is removed before aggregation.
 
A positive incentive mechanism can attract high-quality clients to actively participate. In addition, the MUD-PQFed can accurately identify malicious clients under the condition of ciphertext that mitigates potential privacy leakages posed in plaintext FL setting. Furthermore, the MUD-PQFed is designed for quantized FL, which can greatly reduce the communication overhead of the clients, which allows those with limited bandwidth to contribute.

\subsection{Communication Overhead}
As for the communicational overhead, the MUD-PQFed overhead has a $2\times n$ relationship of $n$ in $d^n$-hypermesh. More specifically, for a $d^n$-hypermesh, each client belongs to $n$ groups at the same time, so $n$ gradient value copies and commitments need to be submitted. Without MUD-PQFed (i.e., in plaintext quantized FL), client only needs to submits a single gradient value. We have evaluated the one-round communication overhead when quantized parameters are in plaintext and ciphertext (with applying MUD-PQFed for privacy preserving aggregation), respectively. In the experiment we take $n_c=4$ with $n=2, d=2$ and consider the quantized parameters. For ternary FL, the overhead without MUD-PQFed and the overhead with MUD-PQFed are 15.6~KB and 62.608~KB, respectively. For binary FL, the overhead without MUD-PQFed and overhead with MUD-PQFed are 26~KB and 104.208~KB, respectively. It can be seen that $4\times = 2\times n$ overhead increase is held with $n=2$.

Because for $n_c$=4, clients are divided according to the $2^2$-hypermesh. Each client belongs to two groups at the same time when MUD-PQFed is used, so each client needs to submit the gradients and commitments to two groups at the same time. In contrast, each client only needs to submit one gradient when MUD-PQFed is not used. Notably, such a $2\times n$ (i.e., 8 given $n=4$) overhead is still magnitude smaller than the overhead incurred by applying HE, which can be hundreds of times overhead increase compared to the plaintext counterpart~\cite{popa2011cryptdb,fang2021privacy}.

\section{Conclusion} \label{sec:conclusion}
 
Though the lightweight cryptographic (in particular, the secret sharing enabled MPC) based privacy-preserving FL can efficiently prevent privacy leakage, the security threats in this setting are not well elucidated. This work is the first study showing the trivialness of conducting model corruption attacks against the communication-efficient quantized FL. We then devise the MUD-PQFed to accurately detect malicious clients to eliminate such attacks. Extensive experiments have validated the efficacy and effectiveness of our MUD-PQFed.

\bibliographystyle{IEEEtran}
\bibliography{Reference}

\vspace{-1.0cm}

\end{document}